%% file: OCTMIB_arxiv.tex
\documentclass[twoside,leqno,twocolumn]{article}
\usepackage{xspace}
\usepackage{flushend}
\usepackage{enumitem}
\usepackage{xcolor}
\usepackage{graphicx,epstopdf}
\usepackage{hyperref}

\usepackage{tabularx,booktabs}
\usepackage{multirow}

\usepackage{algorithm}
\usepackage[noend]{algpseudocode}
\usepackage{amsmath,amsthm,amssymb}

\usepackage{thmtools}
\usepackage{thm-restate}

\newtheorem{theorem}{Theorem}[section]

\usepackage{setspace}
\usepackage{enumitem}

\newenvironment{tight_item}{
\begin{itemize}[noitemsep,topsep=0pt,partopsep=0pt, parsep=0pt, leftmargin=1.2em]
  \renewcommand{\labelitemi}{\raise .5ex\hbox{\tiny$\bullet$}}
}{\end{itemize}}

\usepackage{tikz}
\usetikzlibrary{arrows,automata,positioning,shapes}
\usetikzlibrary{decorations.pathreplacing,calc}

\DeclareMathOperator*{\argmin}{argmin}

\input{./notation-defs.tex}

\usepackage{xspace}
\usepackage{xargs}
\usepackage{xifthen}
\usepackage{framed}

\newenvironment{tightcenter}
 {\parskip=0pt\par\nopagebreak\centering}
 {\par\noindent\ignorespacesafterend}

\usepackage{ctable}
\newlength{\RoundedBoxWidth}
\newsavebox{\GrayRoundedBox}
\newenvironment{GrayBox}[1]%
{\setlength{\RoundedBoxWidth}{\linewidth-4.5ex}
\def\boxheading{#1}
\begin{lrbox}{\GrayRoundedBox}
\begin{minipage}{\RoundedBoxWidth}%
}{%
\end{minipage}
\end{lrbox}%
\begin{tightcenter}%
\begin{tikzpicture}%
\node(Text)[draw=black!20,fill=white,rounded corners,%
inner sep=2ex,text width=\RoundedBoxWidth]%
{\usebox{\GrayRoundedBox}};
\coordinate(x) at (current bounding box.north west);
\node [draw=white,rectangle,inner sep=3pt,anchor=north west,fill=white]
at ($(x)+(10.5pt,.75em)$) {\boxheading};
\end{tikzpicture}
\end{tightcenter}\vspace{0pt}%
\ignorespacesafterend
}

\newenvironment{problem}[1]{\noindent\ignorespaces%
\FrameSep=6pt%
\parindent=0pt%
\vspace*{.5em}
\begin{GrayBox}{\textsc{#1}}%
\newcommand\Input{Input:}%
\newcommand\Prob{Output:}%
\begin{tabular*}{\columnwidth}{@{\hspace{.25em}} >{\itshape} p{1.1cm} p{0.8\columnwidth} @{}}%
}{
\end{tabular*}%
\end{GrayBox}%
\vspace*{-.5em}
\ignorespacesafterend
}

\usepackage{graphicx}

\begin{document}

\title{\Large Mining Maximal Induced Bicliques using Odd Cycle Transversals
\thanks{This work was supported by the Gordon \& Betty Moore Foundation's Data-Driven Discovery Initiative
    under Grant GBMF4560 to Blair D. Sullivan.}}
\author{Kyle Kloster \thanks{North Carolina State University.} \\
Blair D. Sullivan\footnotemark[2] \\
\and
Andrew van der Poel\footnotemark[2]}
\date{}

\maketitle

\begin{abstract}
    \small\baselineskip=9pt

    \input{./abstract.tex}

\end{abstract}

\section{Introduction}\label{sec:introduction}

\input{./introduction.tex}

\section{Preliminaries}

\input{./background.tex}

\section{Lexicographic Enumeration}\label{sec:dias}
\input{./appendix-dias.tex}

\section{Algorithms}\label{sec:algorithms}

\input{./methods.tex}

\section{Theoretical Guarantees}\label{sec:theory}

\input{./theory.tex}

\section{Empirical Evaluation}\label{sec:experiments}

\input{./evaluation.tex}

\section{Conclusions}

\input{./conclusions.tex}


\input{./bibliography.tex}
\clearpage
\appendix
\part*{Appendices}

\section{Extremal Case for Zhang et al.}\label{app:zhang}
\input{./appendix-mbea.tex}

\section{\OCTMIB Algorithm}\label{app:octmib}
\input{./appendix-octmib.tex}

\section{\mcb Algorithm}\label{app:mcb}
\input{./appendix-mcb.tex}

\section{Data and Infrastructure}\label{app:data}

\input{./appendix-data.tex}

\end{document}

%% file: notation-defs.tex
\newcommand{\rooti}{r_i}
\newcommand{\bag}{b}
\newcommand{\childi}{c_i}
\newcommand{\bicleft}{A}
\newcommand{\bicright}{B}
\newcommand{\blueprint}{P}
\newcommand{\blueprintspc}{P^*}
\newcommand{\biclique}{\ensuremath{\bicleft \times \bicright}}

\newcommand{\indfrom}{\mathcal{I}}
\newcommand{\ccto}{\mathcal{C}}
\newcommand{\order}{\phi}

\newcommand{\OCTMIB}{\texttt{OCT-MIB}\xspace}
\newcommand{\imbea}{\texttt{iMBEA}\xspace}

\newcommand{\mcb}{\texttt{MCB}\xspace}
\newcommand{\MCB}{\mcb}
\newcommand{\mis}{\texttt{MIS}\xspace}
\newcommand{\dias}{\texttt{Dias}\xspace}
\newcommand{\moddias}{\texttt{LexMIB}\xspace}
\newcommand{\phasestyle}[1]{\textbf{\textit{#1}}}

\newcommand{\AMBpfull}{\textsc{Maximal Induced Biclique Enumeration}\xspace}

\newcommand{\MCBpfull}{\textsc{Maximal Crossing Biclique Enumeration}\xspace}

\newcommand{\algoref}[1]{\hyperref[#1]{Algorithm~\ref*{#1}}}

%% file: abstract.tex
Many common graph data mining tasks take the form of identifying dense subgraphs (e.g. clustering, clique-finding, etc). In biological applications, the natural model for these dense substructures
is often a complete bipartite graph (biclique), and the problem requires enumerating all maximal bicliques (instead of identifying just the largest or densest). The best known algorithm in general graphs is
due to Dias et al., and runs in time $O(M |V|^4)$, where $M$ is the number of \emph{maximal induced bicliques} (MIBs) in the graph. When the graph being searched is itself bipartite, Zhang et al. give a faster algorithm where the time per MIB depends on the number of edges in the graph. In this work, we present a new algorithm for enumerating MIBs in general graphs, whose run time depends on how ``close to bipartite'' the input is. Specifically, the runtime is parameterized by the size $k$ of an odd cycle transversal (OCT), a vertex set whose deletion results in a bipartite graph. Our algorithm runs in time $O( M |V| |E| k^2 3^{k/3} )$, which is an improvement on Dias et al. whenever $ k\leq 3 \log_3 |V|$. We implement our algorithm alongside a variant of Dias et al.'s in open-source C++ code, and experimentally verify that the OCT-based approach is faster in practice on graphs with a wide variety of sizes, densities, and OCT decompositions.\\

\noindent\textbf{Keywords:} bicliques, odd cycle transversal, parameterized algorithms, enumeration, bipartite

%% file: introduction.tex
Bicliques (complete bipartite graphs) naturally arise in many data mining applications, including detecting cyber communities~\cite{KUMAR}, data compression~\cite{AGARWAL}, epidemiology~\cite{MUSHLIN}, artificial intelligence~\cite{WILLE}, and gene co-expression analysis~\cite{KAYTOUE, KAYTOUE2}. In many
settings, the bicliques of interest are \emph{maximal} (not contained in any larger biclique) and/or
\emph{induced} (each side of the bipartition is independent in the host graph), and there
is a large body of literature~\cite{ALEXE,DIAS,EPPSTEIN,LI,MAKINO,MUSHLIN,SANDERSON,ZHANG}
giving algorithms for enumerating all such subgraphs. Many of these approaches make strong structural
assumptions on the host graph; the case when the host graph is bipartite has been particularly
well-studied, and the \imbea algorithm of Zhang et al. has been empirically established to be state-of-the-art~\cite{ZHANG}. In general graphs, the only known non-trivial algorithm for enumerating maximal induced bicliques (MIBs) is that of Dias et al.~\cite{DIAS}.

We consider the problem of efficiently enumerating all MIBs in general graphs.
In particular, we design an algorithm \OCTMIB that extends ideas from \imbea to work on non-bipartite graphs by using an \emph{odd cycle transversal} (OCT set): a set of nodes $O$ such that $G[V\setminus O]$ is bipartite. We prove that our algorithm has runtime $O(M n m \cdot n_O^2 \cdot 3^{n_O/3})$ where $n_O = |O|$, $M$ is the number of MIBs in the graph, and $n$ and $m$ denote the number of vertices and edges in the graph, respectively.
This is asymptotically faster than the approach of Dias et al. whenever the OCT set has size $n_O \leq 3 \log_3 (n)$. Since all graphs have OCT sets (although they can be size $O(n)$, as in cliques),
our algorithm can be run in the general case; its correctness does not require minimality or
optimality of the OCT set.\looseness-1

We also present several additional algorithms which may be of independent interest. The first is \moddias,
a modification of the algorithm of Dias et al.~\cite{DIAS}, which addresses a flaw in the original method
and enumerates all MIBs in time $O(M n^4)$.
The second is for a variant of biclique enumeration where we are only interested in bicliques with one part inside a specified independent set of the host graph. Specifically, given a graph $G = (V,E)$
with $V$ partitioned into $X\cup Y$ with $X$ an independent set, we wish to enumerate all \emph{maximal crossing bicliques} (MCBs) $A \times B$ with $A \subseteq X$, $B \subseteq Y$. We give an algorithm \mcb which enumerates all maximal crossing bicliques in time $O(|X||Y|m)$ per MCB.

Further, we implement both \OCTMIB and \moddias in open source C++ code, and evaluate their performance on a suite of synthetic graphs with known OCT decompositions. Our experiments show that \OCTMIB is generally at least an order of magnitude faster than \moddias, and verify that in practice,
our parameterized approach yields the best performance even when the distance to bipartite (OCT set size) is $\Omega(\log(n))$, far exceeding the constant values typically required by such algorithms.

We begin with preliminaries and a brief discussion of related work, then describe the flaw from the original method in~\cite{DIAS} along with our corrected version in Section~\ref{sec:dias}. We outline our
main algorithm \OCTMIB in Section~\ref{sec:algorithms}, and state the correctness and complexity for \OCTMIB in Section~\ref{sec:theory}. Finally, we present our experimental evaluation in Section~\ref{sec:experiments}.

In Appendix~\ref{app:zhang} we begin by identifying a typo in the runtime of \imbea from~\cite{ZHANG} and show that its runtime is $O(Mmn)$, where $M$ is the number of MIBs. We include detailed descriptions of the mechanics of \OCTMIB and \mcb, along with proofs of their correctness and runtime in Appendices~\ref{app:octmib} and~\ref{app:mcb}, respectively. 
In Appendix~\ref{app:data} we describe our random graph generator and the hardware we use for our experiments, along with an additional experiment comparing \OCTMIB and \moddias.\looseness-1

%% file: background.tex
\subsection{Related work} \label{sec:related-work}

The complexity of finding bicliques is well-studied, beginning with the results
of Garey and Johnson~\cite{GAREY} which establish that in bipartite graphs, finding the largest balanced biclique is NP-hard but the largest biclique (number of vertices) can be found in polynomial time. Finding the biclique with the largest number of edges was shown to be NP-complete in general graphs~\cite{YANNAKAKIS}, but the case of bipartite graphs remained open for many years. Several variants (including the weighted version) were proven NP-complete in~\cite{DAWANDE}, and in 2000, Peeters finally resolved the problem, proving the edge maximization variant is NP-complete in bipartite graphs~\cite{PEETERS}. Particularly relevant to the mining setting, Kuznetsov showed that enumerating maximal bicliques in a bipartite graph is \#P-complete~\cite{KUZNETSOV}, the NP-completeness analogue for counting problems~\cite{VALIANT}.

For the problem of enumerating maximal \emph{induced} bicliques, the best known algorithm in general graphs is due to Dias et al.~\cite{DIAS}; in the non-induced setting, other approaches include a consensus algorithm~\cite{ALEXE}, an efficient algorithm for small arboricity~\cite{EPPSTEIN}, and a general framework for enumerating maximal cliques and bicliques~\cite{GELY}.
We note that, as described, the method in~\cite{DIAS} may fail to enumerate all MIBs; we describe a graph eliciting this behaviour, along with a modified algorithm (\moddias) with proof of correctness in Section~\ref{sec:dias}.
We note that our correction increases the runtime of the approach from $O(n^3)$ to $O(n^4)$ per MIB.\looseness-1

There has also been significant work on enumerating MIBs in bipartite graphs.
We note that since all bicliques in a bipartite graph are necessarily induced, non-induced solvers for general graphs (such as~\cite{ALEXE}) can be applied, and have been quite competitive.
The best known approach, however, is an algorithm due to Zhang et al.~\cite{ZHANG} that directly exploits the bipartite structure\footnote{Due to a typo in their runtime (see Section~\ref{app:zhang}), the worst-case complexity of this algorithm is not an improvement on the  $O(n^2)$ time per MIB of several other approaches. However, in practice, the experimental results in~\cite{ZHANG} support this being faster than~\cite{ALEXE}.}.
This algorithm works by iterating over one partition of the bipartite graph and iteratively building maximal bicliques. They maintain special subsets of the iterated-over partition for each biclique and leverage the structure implied by bipartite graphs to efficently enumerate all of the maximal bicliques.
Other approaches in the bipartite setting include frequent closed itemset mining~\cite{LI} and transformations to the maximal clique problem~\cite{MAKINO}; faster algorithms are known when a lower bound on the size of bicliques to be enumerated is assumed~\cite{MUSHLIN, SANDERSON}.\looseness-1

In this work, we extend techniques for bipartite graphs to the general setting using
odd cycle transversals, a form of ``near-bipartiteness'' which arises naturally in many applications~\cite{GULPINAR,PANCONESI,SCHROOK}.  Although finding a minimum size
OCT set is NP-hard, the problem of deciding if an OCT set with size $k$ exists is fixed parameter tractable (FPT), with algorithms in~\cite{LOKSHTANOV} and~\cite{IWATA} running in times $O(3^k k m n)$ and $O(4^k n)$, respectively.
Although it is possible for graphs to have minimum OCT sets as large as $O(n)$ (as in cliques, for example),
in practice many graphs from common problem domains have OCT sets of size at most 40~\cite{HUFFNER}.

Other algorithms for OCT include a $O(\sqrt{log(n)})$-approximation~\cite{AGARWAL2},  a randomized polynomial kernelization algorithm based on matroids~\cite{KRATSCH}, and a subexponential algorithm for planar graphs~\cite{LOKSHTANOV2}.
Since our algorithm only requires a valid OCT set (not a minimal or optimal one), any of these approaches or one of several high-performing heuristics may be used to pre-process the data. Recent implementations~\cite{GOODRICH} of a heuristic ensemble alongside algorithms from ~\cite{AKIBA,HUFFNER} alleviate concerns about finding an OCT decomposition creating a barrier to usability for our algorithm.

\subsection{Notation and terminology}

Let $G = (V,E)$ be a graph. We denote $n = |V|$ and $m = |E|$, and use $N(v)$ to represent the neighborhood of a node $v \in V$. An independent set $T$ in $G$ is a \emph{maximal independent set} (MIS) if $T$ is not contained in any other independent set of $G$. We use $\indfrom(S)$ to denote all nodes which are independent from a set $S$ and $\ccto(S)$ to denote all nodes which are completely connected to a set $S$.

A biclique $\biclique$ in a graph $G = (V,E)$ consists of disjoint sets $A, B \subset V$ such that every vertex of $A$ is connected to every vertex of $B$.
We say a biclique $\biclique$ is \emph{induced} if both $A$ and $B$ are independent sets in $G$.
A biclique is \emph{maximal} in $G$ if no biclique in $G$ properly contains it.
Given a fixed independent set $X \subseteq V$, we define a \emph{crossing} biclique with respect to $X$ to be an induced biclique $\biclique$ such that $A \subseteq X$ and $B \subseteq V\setminus X$.

If $G$ is bipartite, we write $G[L,R]$, where the vertices are partitioned as $V  = L \cup R$ and we refer to the two partitions as the left and right ``sides'' of the graph.
For a biclique $\biclique$ in $G[L,R]$, by convention we list the ``left'' set first, i.e. $A \subseteq L$ and  $B \subseteq R$.

If $G$ has OCT set $O$, we denote the corresponding OCT decomposition of $G$ by $G[L,R,O]$, where the induced subgraph $G[L,R]$ is bipartite, and called the \emph{bipartite part}. We write
$n_L, n_R,$ and $n_O$ for $|L|, |R|,$ and $|O|$, respectively. We let $n_B = n_L + n_R$.
Given an arbitrary vertex set $T \subset V$ we abbreviate $T^O = T \cap O$ and $T^{L,R} = T \setminus O$.

We present algorithms for enumerating maximal induced bicliques in two settings.
The first setting, \AMBpfull, is our primary focus.
\begin{problem}{\AMBpfull}
\Input & A graph $G = (V,E)$.\\
\Prob  & All maximal induced bicliques in $G$.\\
\end{problem}

The second setting, \MCBpfull arises as a subproblem in our approach to \AMBpfull.
\begin{problem}{\MCBpfull}
\Input & A graph $G = (V,E)$; with $V$ partitioned into $X \cup Y$ s.t. $X$ is independent in $G$.\\
\Prob  & All maximal bicliques $\biclique$ in $G$ where $A \subseteq X$ and $B \subseteq Y$.\\
\end{problem}

%% file: appendix-dias.tex
In this section we identify graphs containing maximal induced bicliques which would not be discovered by the algorithm of Dias et al. as originally stated in~\cite{DIAS}, which we refer to as \dias. We then describe a modified approach, which we prove is guaranteed to output all maximal induced bicliques in lexicographic order\footnote{We believe \dias could also be modified by removing the if conditions on lines 9 and 11 to output all MIBs in non-lexicographic order, with runtime $O(M n^3)$, but did not focus on this alternate strategy.} in time $O(M n^4)$, where $M$ is the number of MIBs in the graph.\looseness-1

For the reader's convenience we have transcribed the pseudo-code of \dias from~\cite{DIAS} in Algorithm \ref{alg:diaspseudo}. Line 13 relies on a subroutine described in the original paper, which finds the lexicographically least biclique containing a given set of nodes in $O(n^2)$ time. For consistency with~\cite{DIAS}, we let $N_j$ be the neighbors of node $j$ and $\bar{N}_j$ the non-neighbors throughout this section.

\begin{algorithm}[!h]
\begin{algorithmic}[1]
\Require Graph $G = (V, E)$, order on V
\Ensure List of all bicliques of $G$ in lexicographic order
\State Find the least biclique $B*$ of $G$
\State $Q \leftarrow \emptyset$
\State insert $B*$ in the queue $Q$
\While {$Q \neq \emptyset$}
	\State find the least biclique $B = X \cup Y$ of $Q$
	\State remove $B$ from $Q$ and output it
	\For {each vertex $j \in V \setminus B$}
		\State $X_j \leftarrow X \cap \{1, \dots, j\}$; $Y_j \leftarrow Y \cap \{1, \dots, j\}$
		\If {$X_j \cap N_j \neq \emptyset$ or $Y_j \cap \bar{N}_j \neq \emptyset$}
			\State $X'_j \leftarrow (X_j \setminus N_j) \cup \{j\}$; $Y'_j \leftarrow Y_j \setminus \bar{N}_j$
			\If {there exists no $\l \in \{1, \dots, j-1\} \setminus B_j$ \hspace*{16mm}such that $X'_j \cup Y'_j \cup \{l\}$ extends to a \hspace*{16mm}biclique of $G$}
				\State find the least biclique $B'$ of $G$ \hspace*{21mm}containing $X'_j \cup Y'_j$, if any
				\If {$B' \neq \emptyset$ and $B' \in Q$}
					\State Include $B'$ in $Q$
				\EndIf

			 \EndIf
		\EndIf
		\State Swap contents of $X_j$ and $Y_j$, repeat lines 9 \hspace*{11mm}to 14 (once per \textbf{for} loop)
	\EndFor
\EndWhile
\end{algorithmic}
\caption {\label{alg:diaspseudo}\dias pseudo-code~\cite{DIAS}.}
\end{algorithm}

We point out that the queue $Q$ needs to be able to recall all MIBs which have been stored in it, which can easily be accomplished by augmenting the data structure without impacting the complexity.
Furthermore, the pseudo-code in Algorithm \ref{alg:diaspseudo} was corrected in line 12 to exclude $j$ from the range of $l$ values.

In the proof of correctness in~\cite{DIAS} they show that for any MIB $B'$ there exists a node $j \in B'$ and another MIB $B$ which does not contain $j$, but does contain both $B'_{j-1} = B' \cap \{1, \dots, j-1\}$ and some $l \in \{1, \dots, j-1\} \setminus B'_{j-1}$. They examine the maximal such $j$ for each $B'$, and consider the iteration where $B$ and $j$ are defined as such (lines 7 and 8). Without loss of generality, assume that $j$ should be added to $X_j$ (as defined in line 9). If $Y_j$ only contains non-neighbors of $j$ and $X_j$ contains non-neighbors of $j$ which are not in $B'$, then no biclique will be found in line 15. Thus the algorithm as written will not produce all of the MIBs, as shown in Figure~\ref{fig:diascounterexample}.

\begin{figure}[h!]
  \centering

\begin{tikzpicture}[>=stealth,shorten >=1pt,auto,node distance=1cm,
  thick,main node/.style={circle,fill=blue!40,draw,font=\sffamily\bfseries},path/.style={rectangle,fill=blue!20,draw,font=\sffamily\bfseries},bunch/.style={ellipse,fill=blue!20,draw,font=\sffamily\bfseries},
red node/.style={circle,fill=red!40,draw,font=\sffamily\bfseries},path/.style={rectangle,fill=blue!20,draw,font=\sffamily\bfseries},bunch/.style={ellipse,fill=blue!20,draw,font=\sffamily\bfseries},
gray node/.style={circle,fill=gray!20,draw,font=\sffamily\bfseries},path/.style={rectangle,fill=blue!20,draw,font=\sffamily\bfseries},bunch/.style={ellipse,fill=blue!20,draw,font=\sffamily\bfseries}]

\node[gray node] (0) {0};
  \node[main node] (1) [right of=0] {1};
  \node[gray node] (2) [right of=1] {2};
    \node[main node] (3) [right of=2] {3};
      \node[gray node] (4) [right of=3] {4};
           \node[red node] (5) [right of=4] {5};

        \draw[every node/.style={font=\sffamily\small}]
     (0) edge [bend left] node {} (2)
    (1) edge node {} (2)
    edge [violet, bend right] node {} (5)
    (2) edge [bend left] node  {} (4)
	edge [bend left] node  {} (5)
    (3) edge [violet, bend right] node  {} (5)
    edge  node {} (4)
 (4) edge node  {} (5);

\end{tikzpicture}

  \caption{\label{fig:diascounterexample}
    An example of a graph where \dias would not find all MIBs. The MIB $C = \{1, 3\} \times \{5\}$ would not be found.
    In the notation of the algorithm $j = 3$ and $B = \{0, 1, 4\} \times \{2\}$ should yield $C$. Note that $ \{0, 1, 4\} \times \{2\}$ is the lexicographically least biclique. Consider when $B = \{0, 1, 4\} \times \{2\}$ and $j = 3$, $X_j = \{0,1\}$ and no biclique is found as $\{0,1,3\}$ have no common neighbors. When $X_j = \{2\}$ the biclique $\{2,3\} \times \{4\}$ is found. Thus $C$ is not added to $Q$. We note that when  $B = \{0, 1, 4\} \times \{2\}$ and $j=5$, we would add the bicliques $\{1,4\} \times \{5\}$ and $\{0,5\} \times \{2\}$ to $Q$. When $B =\{0, 5\} \times \{2\}$ (line 5), no new bicliques are added to $Q$ and then the next least biclique output in line 5 would not be $C$.
  }

\end{figure}
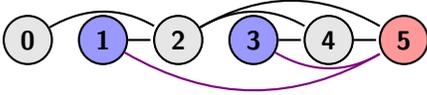

\subsection{Modified \dias}

We now describe \moddias and show that it finds all MIBs in lexicographic order.
The first issue in \dias arises when in line 8, $Y_j$ is empty and $X_j$ does not contain any neighbors of $j$. In this case, we fail to satisfy the \textbf{if}-condition in line 9.
The problem is resolved by appending an additional \textbf{or}-condition in line 9:
``\textbf{or} $|Y_j| = 0$''.

The other problematic case is when $Y'_j$ is empty and $X_j$ contains non-neighbors of $j$ which are not in the lexicographically next MIB (e.g. node $0$ for $\{1,3\} \times \{5\}$ in Figure \ref{fig:diascounterexample}). In this case, $B' = \emptyset$ in line 13,
so we add a new \textbf{else-if} clause to the conditional:
``\textbf{else if} $B' = \emptyset$ and $Y'_j = \emptyset$ \textbf{then}
\textbf{for all} $v \in N_j$ find the least biclique containing $(N(v) \cap X'_j, \{v\})$''.

We now argue that \moddias will find any biclique $B'$ missed by \dias in the above manner. Let $S = B' \cap X'_j$. Let $v_1$ be the least node in $B'$ which shares an edge with $j$.
Consider the iteration of the above process where $v = v_1$. Clearly all of $S$ is contained in $N(v_1) \cap X'_j$. Furthermore due to the maximality of $j$, the original argument from~\cite{DIAS} can be used to show that finding the least biclique containing $(N(v_1) \cap X'_j, \{v\})$ returns $B'$. In Figure \ref{fig:diascounterexample} where biclique $C$ is missed, when $j = 3$ and $B = \{0, 1, 4\} \times \{2\}$, $v_1 = 5$
and $N(v_1) \cap X'_j = \{1, 3\}$, which returns this previously missing biclique $C = \{1, 3\} \times \{5\}$.

This augmentation increases the delay time of the algorithm to be $O(n^4)$ since finding the least biclique may be called $O(n)$ times within the \textbf{for} loop at line 7, which itself has $O(n)$ iterations. Note that this addition does not alter their argument for the bicliques being output in lexicographic order.

%% file: methods.tex
Here we describe two novel algorithms for biclique enumeration: \OCTMIB which enumerates all MIBs in a general graph $G$ by using an OCT decomposition $G[L,R,O]$ to drive a divide and conquer approach,
and \mcb which computes all MIBs with one partition in a designated independent set.
In \OCTMIB, removing an OCT set from $G$ enables use of efficient methods for enumerating MIBs in the bipartite setting, such as Zhang et al.'s \imbea algorithm.
Then a given MIB, $A \times B$, in the bipartite graph $G[L,R]$ can be checked for maximality in $G$ by attempting to add vertices from $O$ to $A \times B$.

Each MIB not found in $G[L,R]$ necessarily contains at least one vertex $v \in O$, allowing us to
 enumerate them by iterating over each vertex $v \in O$ and identifying all MIBs that contain $v$.
This process requires careful bookkeeping that we organize using the observation that in all MIBs containing a given vertex $v$, one side of the biclique must be an independent set completely connected to $v$.
Hence, we proceed with constructing ``seed'' bicliques from independent sets in $N(v)$.
We then ``grow'' these into maximal bicliques by adding vertices from an MIS in $O$ containing $v$.

Next we will
provide more detailed algorithm outlines for both \OCTMIB and \mcb.

\begin{algorithm}[!h]
\begin{algorithmic}[1]
\Require Graph $G = (V, E)$, partitioning of $V$ into $L, R, O$, where $L$ and $R$ are independent sets
\Ensure $\mathcal{M}$, all maximal induced bicliques of $G$
\State $\mathcal{M} \leftarrow \emptyset$

	\State $B_1 \leftarrow$ maximal bicliques in $G[L,R]$
	\For {$b \in B_1$}
		\If {$b$ is maximal in $G$}
			\State  add $b$ to $\mathcal{M}$
		\EndIf
	\EndFor
	\State $I_O = \mis(O)$
	\For {$S \in I_O$}
		\State $T \leftarrow \emptyset$ \Comment{to hold bags}
		\State Fix $\phi$ an order of $S$

		\For {$v \in S$}
			\State $B_2 \leftarrow$ unique initial bicliques via \mcb \& \mis
			\For {$b \in B_2$}
				\If {\texttt{not-future-max}($b$, $G$)}
					\State remove $b$ from $B_2$
				\EndIf
				\If {\texttt{maximal}($b$, $G$)}
					\State add $b$ to $\mathcal{M}$ (and keep $b$ in $B_2$)
				\EndIf
			\EndFor

		\If {$B_2$ is not empty}
			\State add bag $B_2$ to $T$
		\EndIf
		\EndFor
		\While {$T$ is not empty}
			\State Pick $\mathcal{B} \in T$ and let $T \leftarrow T \setminus \mathcal{B}$
			\For {$v \in S$: $\phi(v) > \phi(w) \, \forall\, w \in S \cap \mathcal{B}$}
				\State $B_3 \leftarrow \emptyset$
				\For {biclique $b \in \mathcal{B}$}
					\State $b* \leftarrow$ expand $b$ with $v$
					\If {\texttt{not-future-max}($b*$, $G$)}
						\State \textbf{continue}
					\EndIf
					\State $B_3 \leftarrow B_3 \cup \{b*\}$
					\If {\texttt{maximal}($b*$, $G$)}
						\State add $b*$ to $\mathcal{M}$ (and keep $b*$ in $B_3$)
					\EndIf
				\EndFor
				\If {$B_3$ is not empty}
					\State  add $B_3$ to $T$
				\EndIf
			\EndFor
		\EndWhile
	\EndFor
	\end{algorithmic}
\caption[]{\OCTMIB pseudo-code.\\Key phases: Bipartite (Lines 2--5), Initialization (Lines 10--18), and Expansion (Lines 19--31).}\label{losemibpseudo}
\end{algorithm}

\vspace{1em}

\begin{figure*}
    \includegraphics[width=0.5\textwidth]{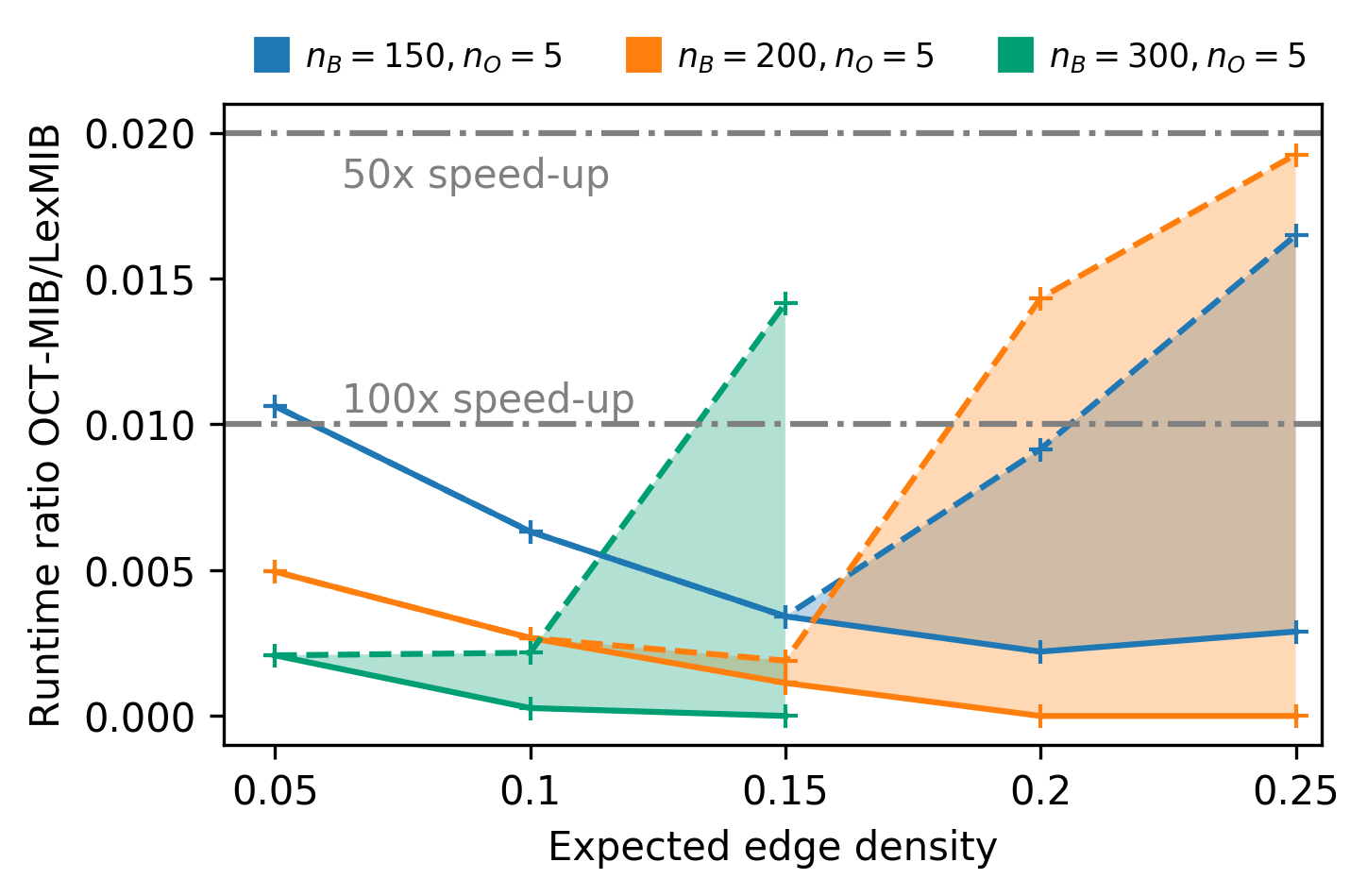}%
    \includegraphics[width=0.5\textwidth]{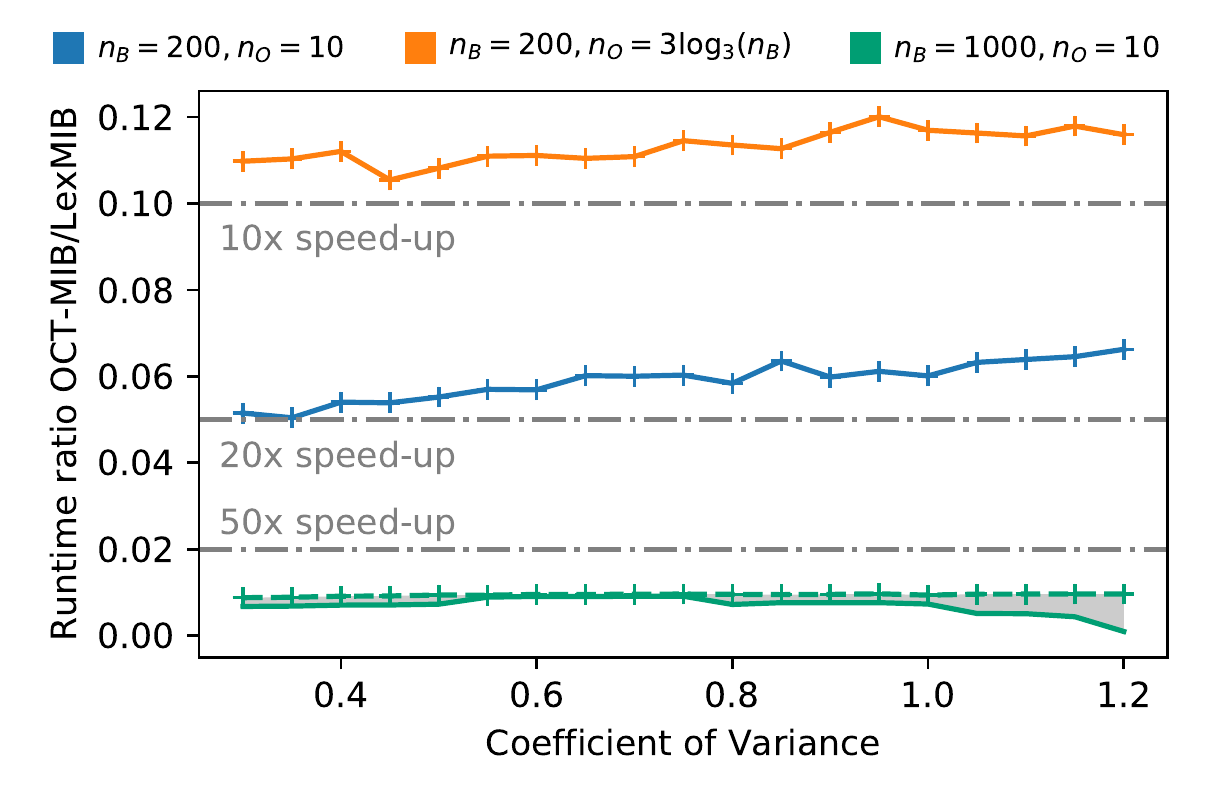}%
    \caption{\label{fig:ed-cv-envelope}
        Runtime ratios of \OCTMIB to \moddias with varying edge densities \textit{(Left)} and coefficient of variance \textit{(Right)}.
        Each curve corresponds to a fixed $n_B$ and $n_O$; for settings where \OCTMIB completed but \moddias timed out on some instances, we display both upper (dashed, using 3600s for \moddias) and lower (solid, using infinity for \moddias) bounds on the ratio.
        For edge density, each point is averaged over 10 randomly generated instances (5 with OCT-OCT density equal to $\bar{d}$ and 5 with it fixed to $5\%$, so as to ensure the results are not an artifact of the OCT-OCT density; when $\bar{d}=5\%$ there are only 5 instances); note that both algorithms timed out on all instances with $n_B$ = $300$ when edge density was greater than $0.15$. For cv, points are an average over 5 instances.
    }
\end{figure*}

\subsection{Algorithm outline}\label{sec:outline}

In both \MCB and \OCTMIB we take an independent set $S$ and build bicliques $\biclique$ such that
$S_I := S \cap (\biclique) \subseteq \bicleft$. We let $\bar{S_I} := \bicleft \setminus S_I$.
During the \phasestyle{initialization phase} we find bicliques of the form $\biclique$ where $\bicleft$ contains exactly one node from $S$. During this phase, we enumerate MISs in a subgraph using \mis~\cite{TSUKIYAMA};
in \OCTMIB, we also utilize \mcb.

We then ``grow'' the bicliques found during initialization in the \phasestyle{expansion phase} by repeatedly adding a node $w \in S \setminus \bicleft$ to $S_I$ and removing nodes from $\bicright$ and $\bar{S_I}$ to ensure $\biclique$ is still an induced biclique. We refer to this process as \emph{expanding with $w$}.
In \mcb we let $S = X$, and in \OCTMIB we let each MIS in $O$ be $S$ once. While \mcb only contains \phasestyle{initialization} and \phasestyle{expansion phases}, \OCTMIB also contains the \phasestyle{bipartite phase} alluded to earlier, where the MIBs which are completely contained in $L \cup R$ are found. (Line 2 of \algoref{losemibpseudo}; this can be completed using e.g. \imbea).

We employ an ordering $\order$ on the vertices of $S$ to limit the amount of redundant expansions we make.
We say $\biclique$ is \emph{near-maximal} if $\biclique$ is maximal with respect to $V \setminus (S \setminus S_I)$ and there is some set of nodes $S' \subseteq S \setminus \bicleft$ such that $(\bicleft \cup S') \times \bicright$ is a maximal biclique. We allow $S' = \emptyset$ so maximal bicliques are also near-maximal.
We refer to near-maximal bicliques where all of the missing nodes occur later in $\order$ as \emph{future-maximal}, and we only expand on future-maximal bicliques (discarding all others).
We group bicliques which contain the same set of $S$-nodes in \emph{bags}. During the \phasestyle{expansion phase} we iterate over the bags, expanding their set of bicliques with all $x \in S$ such that
$\order(x) > \order(y)$ for all $y \in S$ occuring in some biclique in the bag. For each $S$-node we expand with, we create a new bag to hold the new future-maximal bicliques.
The general approach of \OCTMIB is outlined in \algoref{losemibpseudo}. The pseudo-code for \mcb is similar but does not include a \phasestyle{bipartite phase} (lines 2 to 5) and only uses \mis in initialization.

%% file: theory.tex
We now state the correctness and asymptotic complexity of \OCTMIB.
The corresponding proofs rely on details from Appendix~\ref{sec:octmib} and can be seen in Appendix \ref{sec:octmibproofs}.\looseness-1

\begin{restatable}{thm}{octmibcorrect}
\label{thm:octmibcorrect}
\OCTMIB finds all maximal induced bicliques.
\end{restatable}

\begin{restatable}{thm}{octmibcomplexity}
\label{thm:octmibcomplexity}
    Given a graph $G$ with OCT decomposition $G[L,R,O]$,
\OCTMIB runs in $O(Mmnn_O^2I_O)$ time, where $M$ is the number of MIBs and $I_O \leq 3^{n_O/3}$ is the number of maximal independent sets in $O$. Its space complexity is $O(Mn)$.
\end{restatable}

%% file: evaluation.tex
In this section we evaluate the performance of \OCTMIB on a suite of synthetic graphs with a variety of sizes, densities, degree distributions, and OCT decomposition structures, to see how various aspects of the graph impact the runtime. We benchmark against \moddias, our modified version of the approach described in~\cite{DIAS} (see Section~\ref{sec:dias}).

\begin{table}[h]
    \centering
{\small
    \begin{tabularx}{0.7\columnwidth}{rrrr}
        \toprule
        $(\bar{d}, b)$ & $n_B$ & $n_O$ & $\bar{M}$ \\
        \midrule
        \multirow{2}{*}{(5\%, 1:10)} & 1000 & - & 3717.8 \\
                                    & 1000 & 20 & 16631.0 \\
        \midrule
        \multirow{2}{*}{(5\% 1:1)} & 1000 & - & 50424.4 \\
                                    & 600 & 10 & 962305.8 \\
        \midrule
        \multirow{2}{*}{(10\%, 1:1)} & 600 & - & 86239.0 \\
                                      & 300 & 5 & 185760.4 \\
        \bottomrule
    \end{tabularx}
}
    \caption{\label{tab:mib-count}
       Comparison of the average number of MIBs $\bar{M}$ in bipartite and near-bipartite graphs with
       equivalent expected edge density $\bar{d}$ and balance $b$.
    }
\end{table}

\subsection{Experimental setup and data}
We implemented \OCTMIB and \moddias\footnote{
The experiments described in this manuscript were run using an early implementation of \moddias
which output bicliques according to a fixed ordering different from that described in~\cite{DIAS}.
The implementation available at~\cite{HORTON} removes this inconsistency.
}
in C++;
we note that no prior implementation of \dias was available.
All code is open source under a BSD 3-clause license and publicly available at~\cite{HORTON}.
Hardware specifications are in Appendix~\ref{app:data}.

\begin{figure*}
    \includegraphics[width=0.5\textwidth]{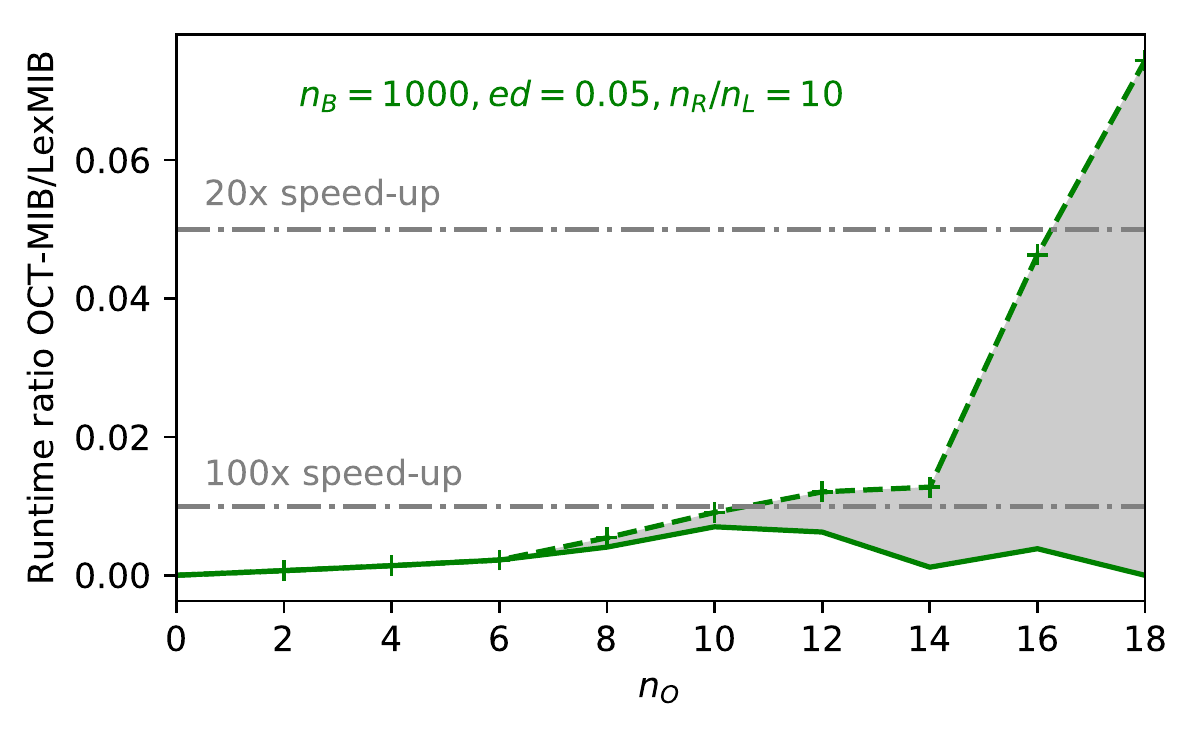}%
    \includegraphics[width=0.5\textwidth]{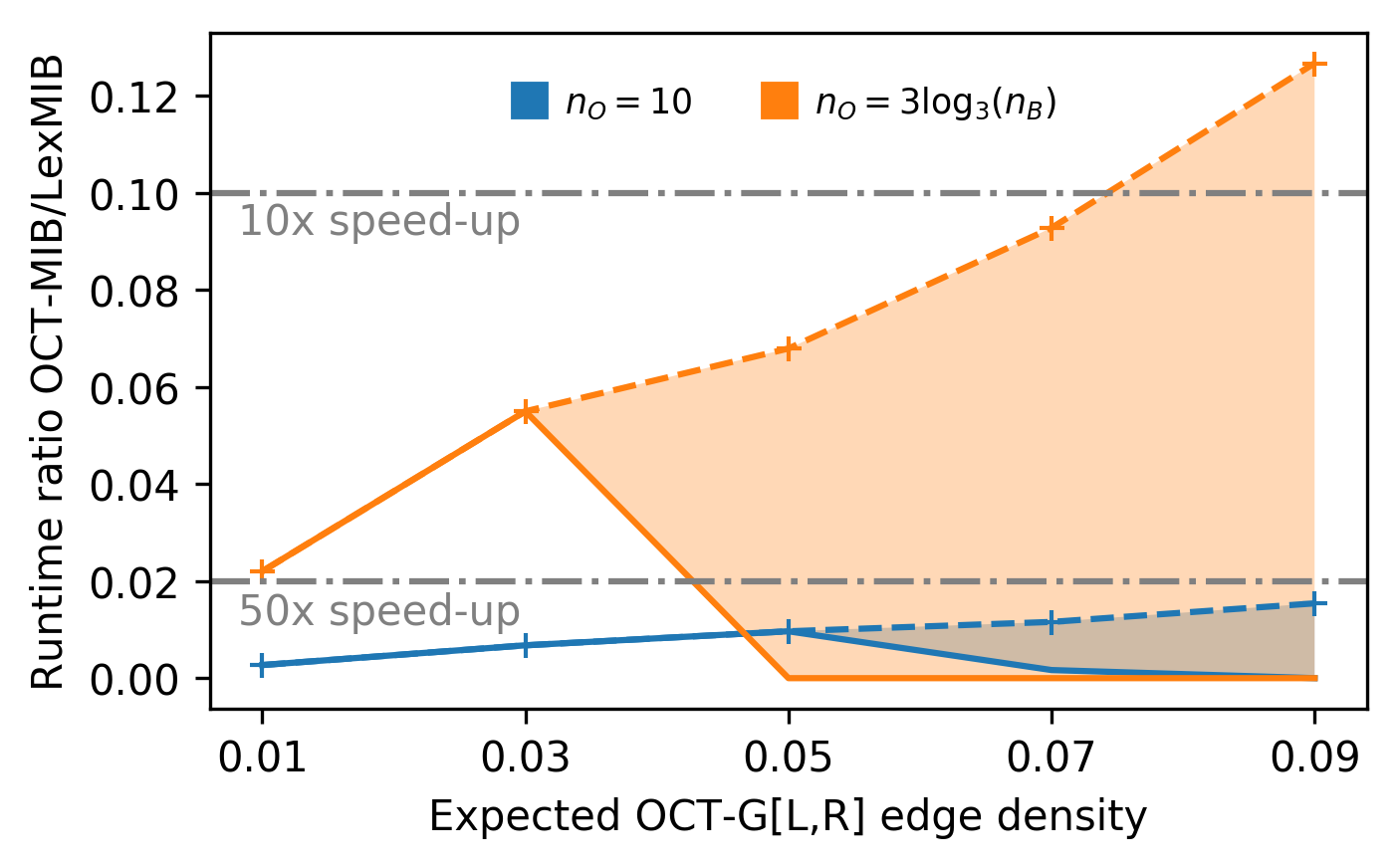}%
    \caption{\label{fig:octsize-octgdensity}
        Runtime ratios of \OCTMIB to \moddias with varying OCT sizes \textit{(Left)} and expected OCT-$G[L,R]$ edge densities \textit{(Right)}. Each data point represents the average over 5 random instances with $n_B = 1000$ and $b = 1:10$.
        For settings where \OCTMIB completed but \moddias timed out on some instances, we display both upper (dashed, using 3600s for \moddias) and lower (solid, using infinity for \moddias) bounds on the ratio. In the OCT-$G[L,R]$ edge density experiment, the expected edge density is $5\%$
        within $O$ and between $L$ and $R$.
    }
\end{figure*}

\noindent\paragraph{Data.}
Our datasets are generated using a modified version of the random graph generator of Zhang et al.~\cite{ZHANG} that augments the random bipartite graph to have OCT sets of known size. The generator allows a user to specify the sizes of $L$, $R$, and $O$ ($n_L$, $n_R$, and $n_O$), the expected edge densities between $L$ and $R$, $O$ and $\{L,R\}$, and within $O$, and control the coefficient of variation (cv; the standard deviation divided by the mean) of the expected number of neighbors in the larger partition over the smaller partition and in $\{L, R\}$ over the OCT set. The generator identifies the set $O$ when it creates each graph; these sets $O$ are used in \OCTMIB, but the techniques mentioned in Section~\ref{sec:related-work} could also be used to find a valid OCT set.
For additional details on the generator and parameter settings see Appendix~\ref{app:data}.

Unless otherwise specified, the following default parameters are used: expected edge density $\bar{d} = 5\%$, $cv=0.5$, $n_B = 1000$ and $b = $1:10; additionally, the edge density between $O$ and $L \cup R$ is the same as that between $L$ and $R$. For all experiments where it was appropriate we tested two OCT sizes, $n_O=10$ and $n_O=\max(5, 3\log_3(n))$. Our default timeout was one hour (3600s).

We note that many of our graphs are significantly smaller than those used in~\cite{ZHANG}, due to an explosion in the number of MIBs in even slightly non-bipartite instances (see Table~\ref{tab:mib-count}). Our graph corpus was designed to include instances with approximately the
same number of MIBs as those in~\cite{ZHANG} for all experiments evaluating variation in the
bipartite subgraph.

\subsection{Comparison of \OCTMIB and \moddias}
We first measured the impact of the heterogeneity of the degree distribution by varying the cv between $L$ and $R$ from $0.3$ to $1.2$ in steps of $0.05$ (cv between $O$ and $\{L, R\}$ is still 0.5). As seen in the right panel in Figure~\ref{fig:ed-cv-envelope},
the ratio of the runtime of \OCTMIB to that of \moddias generally decreases as cv increases when $n_B=200$, but is at least consistent, and possibly increases when $n_B=1000$.

In order to evaluate the effect of edge density, we restricted our attention to graphs with 1:1 balance (as in~\cite{ZHANG}), where $n_B \in \{150, 200, 300\}$ and uniform expected edge density ranges from $0.05$ to $0.25$ in steps of $0.05$. Internal density within $O$ was either set to match edge density or fixed to be $0.05$. As seen at left in Figure~\ref{fig:ed-cv-envelope}, the ratio of the runtime of \OCTMIB to that of \moddias decreases as the expected density increases to $0.15$, after which it is likely to continue to decrease but can not be determined due to \moddias timing out.

Finally, as was done in~\cite{ZHANG}, we tested the effect of the size $n_B$ and balance $b$ of the bipartite subgraph on runtime. This was particularly challenging due to the significant increase in the number of MIBs at lower balance ratios (see Table~\ref{tab:mib-count}). Results and discussion are in Appendix~\ref{app:data} (Figure~\ref{fig:balance-small}); the effect of making the bipartite subgraph more imbalanced was similar for both algorithms.

\subsection{Varying OCT structures}\label{sec:oct-exploration}

This set of experiments explores how \OCTMIB performs on graphs with widely varying OCT sizes and
densities.

We tested $n_O$ by varying it from $0$ to $18$ in steps of size $2$ (left panel of Figure~\ref{fig:octsize-octgdensity}), noting that the runtime of \OCTMIB tends to increase relative to that of \moddias as $n_o$ increases. We examined the impact of the density between $O$ and $\{L, R\}$ by varying it from $0.01$ to $0.09$ in steps of size $0.02$ (right panel of Figure \ref{fig:octsize-octgdensity}), observing that the runtime of \OCTMIB slightly increases relative to that of \moddias as the expected density increases. In the latter we let the cv between $O$ and $\{L, R\}$ be 0.

We also looked at graphs where the structure of $O$ was optimal for \OCTMIB (independent) and adversarial (perfect matching) with respect to the number of MISs in $O$.
For the best case we fixed $n_B = 10,000$ and the balance at 1:10, while setting $n_O \in \{5, 10, 25\}$, and edge density to $0.005$ and $0.01$. For this experiment, we increased the timeout to 7200 seconds.
Table~\ref{tab:best-case} shows the comparatively modest growth in runtime as the independent OCT set is increased from 5 to 25.

\begin{table}
    \centering
{\small
    \begin{tabularx}{0.5\linewidth}{r|rr}
        \toprule
        $n_O$ & {$\bar{d} = 0.5\%$}   & {$\bar{d} = 1\%$} \\
        \midrule
        5  &  312.7 &  1345.8  \\
        10 & 597.8 & 2828.1\\
	25 & 1552.1 & 6683.9\\
        \bottomrule
    \end{tabularx}
}
    \caption{\label{tab:best-case}
        Runtimes for \OCTMIB on graphs with $n_b=10000$ and independent sets $O$. Each entry represents the average runtime (seconds) over 5 random graphs.
    }
\end{table}

\begin{figure*}[t!]
    \includegraphics [width=0.5\textwidth]{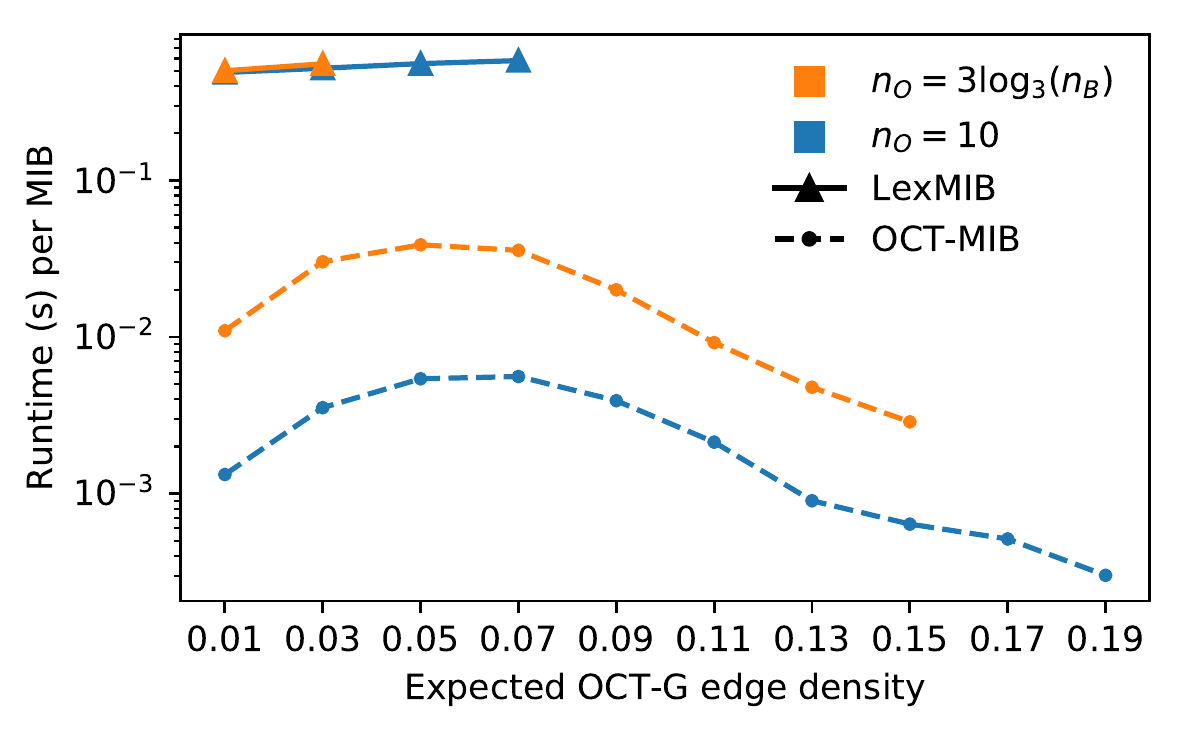}%
    \includegraphics [width=0.5\textwidth]{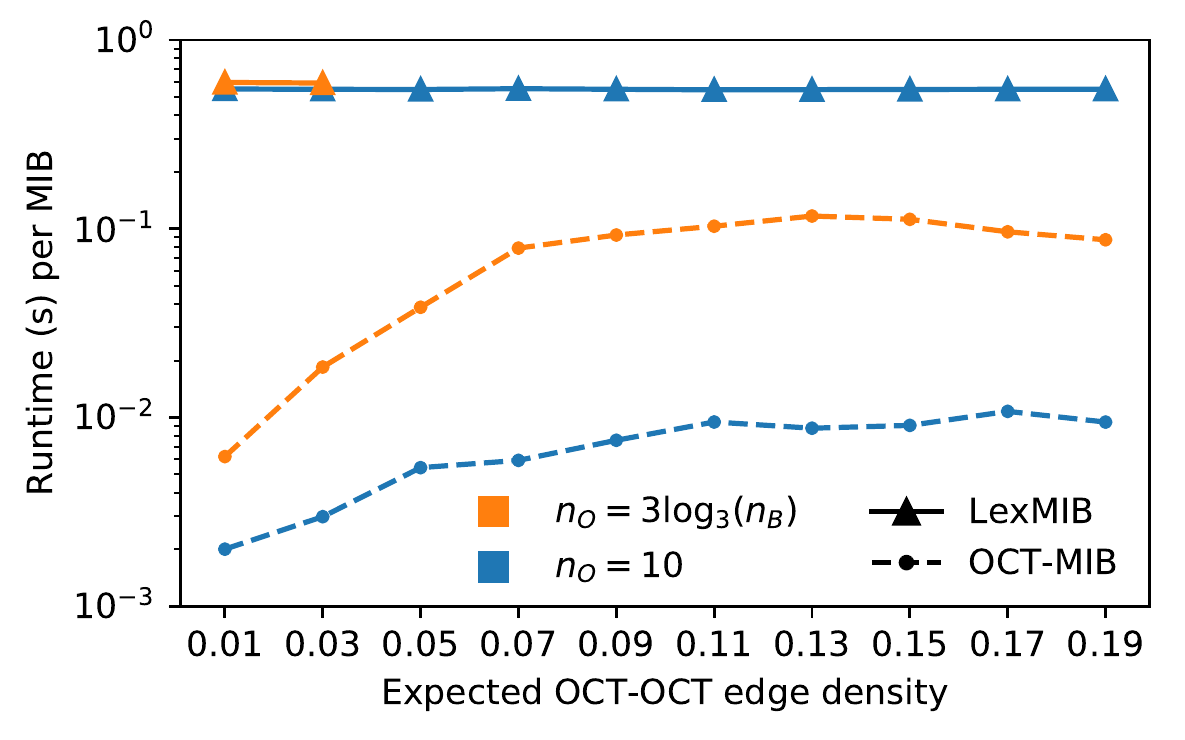}%
    \caption{\label{fig:permib}
    Evaluation of \OCTMIB and \moddias with varying expected OCT-$G[L,R]$ edge densities (\textit{Left}) and varying edge density within the OCT set (\textit{Right}), expressed in terms of average runtime per MIB.    Each data point represents the average over 5 random instances with $n_B = 1000$, $b = 1:10$, and expected edge density $5\%$ in the portion of the graph where edge density is not being varied.  For the setting where we varied the density within the OCT set and $n_O = 3\log_3(n_B)$, \moddias timed out on all instances with density greater than 0.03.
    }
\end{figure*}

For graphs with worst-case OCT sets (perfect matching), we let $n_B$ be $300$ or $600$, $n_O = 3\log_3(n)$, and set the balance to 1:1. As seen in Table~\ref{tab:worst-case}, \OCTMIB outperforms \moddias by nearly an order of magnitude.

To evaluate the effect of the edge density within $O$, we generated graphs where the expected density within the OCT set varied from $0.01$ to $0.19$ in steps of size $0.02$ (see the right panel of Figure~\ref{fig:permib}). In order to best observe the impact of the density within the OCT set we let the cv between the OCT set and $\{L, R\}$ be 0. It does not appear that the density within the OCT set has much impact on \moddias's runtime. This is likely because \moddias's runtime is heavily influenced by the average degree in $G$, which does not increase significantly as the density within $O$ increases.
However the runtime of \OCTMIB increases as the density within $O$ increases, which is not surprising
due to higher density in $O$ creating more MISs in the OCT set.
Furthermore, this shows that there are likely many regimes where the average time spent per MIB is on the order of thousandths of a second.

In the left panel of Figure~\ref{fig:permib}, we present data from the same OCT-$G[L, R]$ edge density experiment as in the right panel of Figure~\ref{fig:octsize-octgdensity}, displayed as average time per MIB.
One might believe that increasing the density between $O$ and $\{L, R\}$ would increase the number of MIBs, so the increase in runtime observed in Figure~\ref{fig:octsize-octgdensity} follows naturally. However average time per MIB actually decreases for \OCTMIB as the density is increased from 0.07 to 0.19 in both settings. We hypothesize that this may be because more MIBs are found in the \phasestyle{expansion phase} as the density increases, due to more OCT nodes being in the MIB on average. Our complexity analysis tells us that finding MIBs in the \phasestyle{expansion phase} should take less time than the \phasestyle{initialization phase}. These results are particularly interesting because when contrasting them with those where the expected edge density within the OCT set is increased (right panel of Figure~\ref{fig:permib}), we observe a converse effect.

Finally, we evaluated both algorithms on graphs with large values of $n_O$ (where \OCTMIB's computational complexity is worse than that of \moddias) and confirmed that in this scenario,
\moddias is the preferable algorithm in both theory and practice. Specifically, on
graphs where $n_B = 200$, \OCTMIB was faster than \moddias when $n_O = 20$, but at $n_O = 40$, it was already an order of magnitude slower. The effect was even more pronounced for a corpus where $n_B = 400$: \moddias had average runtimes of $119s$ and $603s$ when $n_O$ was $40$ and $80$, respectively, whereas \OCTMIB timed out ($3600$s) on 9/10 instances (finishing a single $n_O = 40$ instance in $2409$s).\looseness-1

\begin{table}
    \centering
{\small
    \begin{tabularx}{0.8\linewidth}{rrrr}
        \toprule
        Algorithm & {$n_B = 300$}   & & {$n_B = 600$} \\
        \midrule
        \moddias  &  (74.0, -) &  &  (3420, 4)  \\
        \OCTMIB  &  (19.8, -)  &  &  (634,  -)  \\
        \bottomrule
    \end{tabularx}
}
    \caption{\label{tab:worst-case}
        Runtimes for \OCTMIB and \moddias on adversarially created graphs. Each entry represents the average runtime (seconds) on completed instances and the number of timeouts ($3600s$). For each $n_B$, we used 5 random graphs with $\bar{d} = 5\%$, $b$ = 1:1,
        $n_O = \lfloor 3 \log_3(n_B) \rfloor$, and an OCT set which is a perfect matching.
    }
\end{table}

%% file: conclusions.tex
We present a new algorithm \OCTMIB for enumerating maximal induced bicliques in general graphs, with runtime parameterized by the size of an odd cycle transversal. Additionally, we describe a flaw in the algorithm of Dias et al.~\cite{DIAS}, and give a corrected variant \moddias.

It is particularly noteworthy that \OCTMIB has the best-known complexity for enumerating MIBs even when the parameter is logarithmic in the instance size---far from the constant regime often necessary for FPT approaches to be efficient.

We implement and benchmark both algorithms on a corpus of synthetic graphs, establishing that
in practice, \OCTMIB is typically an order of magnitude faster than \moddias---even when $n_O$ is $O(\log(n))$. We also confirm that \OCTMIB finishes on graphs with over 1,000,000 MIBs in minutes, and enumerates all MIBs in sparse graphs ($\bar{d}=1\%$) where $n=10,000$ and $n_O$ is $O(\log(n))$ in less than two hours.

Our experiments also demonstrate that size may not be the most important feature of an OCT set in determining \OCTMIB's runtime in practice. Although this paper focused on the algorithm's performance
in the setting where an OCT decomposition with given $n_O$ was known, it would be interesting to
further analyze how the edge structure within the OCT set impacts the runtime (for example, we know that
the number of maximal independent sets in $O$ plays a key role in \OCTMIB's execution), and whether OCT sets found by different algorithms/heuristics are more or less advantageous for biclique enumeration.

We also point out that \MCBpfull for which we give an $O(M|X||Y|m)$ algorithm
is a novel problem in its own right and may be well-motivated by applications where large independent sets naturally arise.

Finally, we note that the current implementation of \OCTMIB could be improved by replacing the MIS-enumeration algorithm with that of~\cite{TSUKIYAMA} and the bipartite phase with
the implementation of \imbea used in~\cite{ZHANG}.\looseness-1

%% file: appendix-mbea.tex
The $O(\mathcal{B} m)$ runtime\footnote{ where $\mathcal{B}$ is the number of MIBs found} stated for the algorithm \imbea in~\cite{ZHANG} contains a typo. The corrected runtime is $O(\mathcal{B} n m)$.
In their analysis they bound the number of nodes in their seach tree and derive their complexity by bounding the time spent on each one by $O(m)$. They show that the number of intermediate search tree nodes is at most $\sum_{i=0}^{d-1} (n-1)^i$, which is $O(\mathcal{B})$, and that the number of leaves is at most $(n-1)^d$ which is said to be $O(\mathcal{B})$. However, by the geometric series, the number of leaves can only be bounded by $O(\mathcal{B}n)$, giving an $O(\mathcal{B}n)$ bound on the number of search tree nodes. We now provide an example of a graph where the number of leaves is $\Omega(\mathcal{B}n)$, implying that the number of leaves cannot be $O(\mathcal{B})$ in the general case.

Graphs which are a perfect matching plus an apex to one node from each edge in the matching provide an example of a graph family where this extremal behavior manifests. Note that these graphs are bipartite with $n_a$ nodes in the smaller partition and $n_b=n_a+1$ in the larger. See Figure \ref{fig:mbeafamily} for the instance where $n_a=4$ and $n_b=5$. There are $\mathcal{B}= (n_a + 1)$ MIBs in such a graph. Assume the nodes in the smaller partition are labeled $a_1, \dots, a_{n_a}$. For $1 <i < n_a$, \imbea attempts to expand the biclique containing $a_i$ and its two neighbors with all nodes in $\{a_{i+1}, \dots, a_{n_a}\}$, each of which creates a leaf in the search tree. Thus there are $\sum_{i=2}^{n_a-1}n_a - i = \Omega({n_a}^2)$ leaves created and the number of leaves is $\Omega(\mathcal{B} n_a) = \Omega(\mathcal{B} n)$.

\begin{figure}[h]
\centering
\begin{tikzpicture}[>=stealth,shorten >=1pt,auto,node distance=1.8cm,
  thick,main node/.style={circle,fill=blue!20,draw,font=\sffamily\bfseries},path/.style={rectangle,fill=blue!20,draw,font=\sffamily\bfseries},bunch/.style={ellipse,fill=blue!20,draw,font=\sffamily\bfseries}]

\node[main node] (0) {$a_1$};
  \node[main node] (1) [right of=0] {$a_2$};
  \node[main node] (2) [right of=1] {$a_3$};
    \node[main node] (3) [right of=2] {$a_4$};
      \node[main node] (4) [below of=0] {$b_1$};
           \node[main node] (5) [right of=4] {$b_2$};
                      \node[main node] (6) [right of=5] {$b_3$};
           		\node[main node] (7) [right of=6] {$b_4$};
                      \node[main node] (8) [right of=7] {$b_5$};

        \draw[every node/.style={font=\sffamily\small}]
     (0) edge  node {} (4)
	edge  node {} (8)
    (1) edge  node {} (5)
	edge  node {} (8)
    (2) edge node  {} (6)
	edge  node {} (8)
    (3) edge  node  {} (7)
	edge  node {} (8);

\end{tikzpicture}

  \caption{\label{fig:mbeafamily} An example of a graph in which \imbea would run in time $\Omega(Bmn)$.
  }

\end{figure}
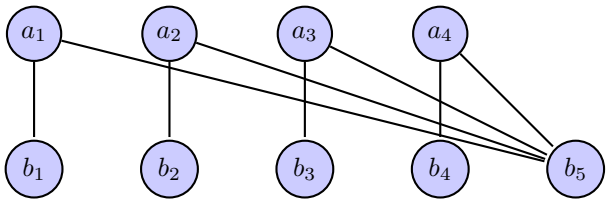

%% file: appendix-octmib.tex
Here we provide the details of the mechanics of \OCTMIB. Please refer to Section \ref{sec:outline} for an outline of \OCTMIB.
Both \OCTMIB and \MCB utilize the \emph{blueprint} data structure, described below.

\noindent\textsc{Definition}. A \textbf{blueprint} is an octuple of sets $(S_I,~IF_I,~CC_I,~S_W,~S_P,~IF_O,~CC_O,~O_{IF})$ and a specified node $next$.
The sets $S_I, S_W, S_P \subseteq S$ satisfy
\begin{tight_item}
\item $S_I \subseteq S$ contains the $S$-nodes in the biclique
\item $S_W := \{x \in S \,|\, \order(x) > \order(w) \forall w \in S_I$\}
\item $S_P := S \setminus (S_I \cup S_W)$.
\end{tight_item}
The sets $IF_I,  IF_O, CC_I, CC_O \subseteq V(G) \setminus S$ with
\begin{tight_item}
\item $IF_I \subseteq \indfrom(S_I)$
\item $IF_O := \indfrom(S_I) \setminus IF_I$
\item $CC_I \subseteq \ccto(S_I)$
\item $CC_O := \ccto(S_I) \setminus CC_I$,
\end{tight_item}
and the set $O_{IF} := (\indfrom(S_I) \cap O) \setminus S$.
Finally,
$$next := \displaystyle{\argmin_{s \in \{S_W \cap \indfrom(IF_I) \cap \ccto(CC_I)\}}\order(s)}.$$
\vspace{0.5em}

By design, $(S_I \cup IF_I) \times CC_I$ is an induced biclique and is \textit{represented}
by the blueprint. We refer to blueprints and the bicliques they represent interchangeably in the description and discussion of the algorithms.

 Within $S$, $S_W$ is the set of nodes that are still candidates to be expanded with, while $S_P$ is the set of  nodes that have been considered to be in a biclique with the current $S_I$ elsewhere. We use $IF_O$ and $CC_O$ to check near-maximality, and $O_{IF}$ to check global maximality in \OCTMIB. The vertex $next$ is used to prevent expansions which produce bicliques which are not future-maximal.
We defer complete descriptions of the algorithmic mechanics to Appendix~\ref{sec:octmib} (\OCTMIB) and Appendix~\ref{sec:mcb} (\mcb), respectively.

\subsection{\OCTMIB Algorithm}\label{sec:octmib}

We begin by explaining checks for various properties of a blueprint, including \texttt{not-future-max}($b$, $G$) and \texttt{maximal}($b,G$). If the node we are expanding with is ordered later than $next$, then we are able to detect that the expansion will not yield a new blueprint. We say that a biclique is \emph{invalid} if $CC_I$ is empty.
A blueprint $b$ is \emph{not future-maximal} if at least one of the following conditions is met; (i)
 $S_P \cap (\ccto(CC_I) \cap \indfrom(IF_I)) \neq \emptyset$,
 (ii) $CC_O \cap (\indfrom(CC_I) \cap \ccto(IF_I)) \neq \emptyset$, or
 (iii) $IF_O \cap (\indfrom(IF_I) \cap \ccto(CC_I)) \neq \emptyset$.
 Note conditions (ii) and (iii) imply $b$ is not near-maximal.
 Finally a blueprint $b$ is \emph{maximal}  $(S_W \cup O_{IF}) \cap (\ccto(CC_I) \cap \indfrom(IF_I)) = \emptyset$. We note that each of these checks can be done in $O(m)$ time.

\noindent\phasestyle{Bipartite phase.}
We run the algorithm from Appendix \ref{app:zhang} on $G[L,R]$, and for each biclique found we check if an OCT node can be added to either side, which can be done in $O(m)$ time. If an OCT node cannot be added then we have found a MIB.

\noindent\phasestyle{Initialization phase.}
Recall that \OCTMIB iterates over MISs in $O$; let $S = s_1, \dots, s_{|S|}$ be the current MIS. For each $s_i$ we create a new  bag $\rooti$ which will hold blueprints. Let $S_I = \{s_i\}$, $S_W = \{s_{i+1}, \dots, s_{|S|}\}$, $S_P = \{s_1, \dots, s_{i-1}\}$ and $O_{IF} = O \setminus (N(s_i) \cup S)$. Now we find candidate blueprints in three rounds, adding the future-maximal ones to $\rooti$.

In the first round, for each MIS $I \subseteq N(s_i)$, we check if any node from $\overline{N}_L(s_i) := L \setminus N(s_i)$ or $\overline{N}_R(s_i) := R \setminus N(s_i)$ is in $\ccto(I)$. If no node from either set is completely connected to $I$, we create a candidate blueprint where $CC_I = I$, $CC_O = N(s_i) \setminus I$, $IF_I = \emptyset$, $IF_O = (L \cup R) \setminus N(v_i)$, and $S_I$, $S_W$, $S_P$, and $O_{IF}$ are defined as above.
Otherwise we continue on to processing the next MIS in $N(s_i)$ and do not create a candidate blueprint for $I$.

In the second round, we run \mcb on the subgraph induced on $\overline{N}_L(s_i) \cup (N(s_i) \setminus L)$ with $X = \overline{N}_L(s_i)$. Then for each returned MCB $\biclique$, we find the set of nodes $\overline{N}_R(s_i)^+ := \overline{N}_R(s_i) \cap (\indfrom(\bicleft) \cap \ccto(\bicright))$. We create a candidate blueprint with $CC_I = \bicright$, $CC_O = N(s_i) \setminus \bicright$, $IF_I = \bicleft \cup \overline{N}_R(s_i)^+$, $IF_O = (L \cup R) \setminus (N(v_i) \cup \bicleft \cup \overline{N}_R(s_i)^+)$,
and $S_I$, $S_W$, $S_P$, and $O_{IF}$ defined as above.

In the third round, we run \mcb on the subgraph induced on $\overline{N}_R(s_i) \cup (N(s_i) \setminus R)$ with $\overline{N}_R(s_i)$ as the designated independent set. For each returned MCB $\biclique$, we
check if $\overline{N}_L(s_i) \cap (\indfrom(\bicleft) \cap \ccto(\bicright)) = \emptyset$.
If so, we create a candidate blueprint where $CC_I = \bicright$, $CC_O = N(s_i) \setminus \bicright$, $IF_I = \bicleft$, $IF_O = (L \cup R) \setminus (N(v_i) \cup \bicleft)$, and $S_I$, $S_W$, $S_P$, and $O_{IF}$ are defined as above. Otherwise, we continue processing the next MCB (without creating a candidate blueprint for $\biclique$).

We check each candidate blueprint for future-maximality, adding it to $\rooti$ if true, and discarding it otherwise.  If the blueprint is maximal we add $(S_I \cup IF_I) \times CC_I$ to the set of maximal bicliques and let $next = \infty$. If the blueprint is not maximal because an OCT node can be added to it we also let $next = \infty$; if it is not maximal because of an $S$-node later in the ordering we let $next$ be the first such node from $S_W$. The blueprint remains in $\rooti$ in either case. As long as $\rooti$ is non-empty we add it to $T$, the set of bags to be processed in the \phasestyle{expansion phase}.

\noindent\phasestyle{Expansion phase.} We now process bags from $T$ until it is empty. We refer to removing a bag $\bag$ from $T$ and expanding on all of the blueprints in $\bag$ with a vertex $v$ as \emph{branching on $\bag$ with $v$}. In \OCTMIB $S_I$, $S_W$, $S_P$, and $O_{IF}$ are the same in all blueprints in a given bag and we branch on a bag with all of the nodes in $S_W$ in the order which matches the order of $S$; $w_1, \dots, w_{|S_W|}$.

Let $w_i$ be the node we are currently branching with and $\childi$ be a bag we create to hold new blueprints formed by expanding with $w_i$. Let $\blueprint = (S_I,~IF_I,~CC_I,~S_W,~S_P,~IF_O,~CC_O,~O_{IF},~next)$ be the blueprint currently being expanded with. If we can detect that the expansion will not yield a new blueprint  we terminate this expansion. Otherwise we create a new blueprint $\blueprint' = (S_I',~IF_I',~CC_I',~S_W',~S_P',~IF_O',~CC_O',~O'_{IF},~next')$ as follows.

Let $S_I' = S_I \cup \{w_i\}$, $IF_I' = IF_I \setminus N(w_i)$, $CC_I' = CC_I \cap N(w_i)$, $S_W' = \{w_{i+1}, \dots, w_{|S_W|}$, $S_P' = S_P \cup \{w_1, \dots, w_{i-1}\}$, $IF_O' = IF_O \setminus N(w_i)$, $CC_O' = CC_O \cap N(w_i)$, and $O'_{IF} = O_{IF} \setminus N(w_i)$. If $(S_I' \cup IF_I') \times CC_I'$
is invalid  or it is not future-maximal,   we terminate the expansion.

We  must consider the case where the tuples $(CC_I', IF_I')$  from different blueprints are identical after expansion. To handle this we maintain a hashtable for the current $w_i$ being branched with, and hash each tuple $(CC_I', IF_I')$  and terminate if we find a conflict. This hashtable can be discarded once we finish branching with the current node.

If the expansion has not been terminated then $\blueprint'$ is future-maximal so we add $\blueprint'$ to $\childi$. If the blueprint is maximal  we add $(S_I' \cup IF_I') \times CC_I'$ to the set of maximal bicliques and let $next' = \infty$. If the blueprint is not maximal solely because of nodes from $O_{IF}'$ we also let $next' = \infty$, but if it is not maximal because of a node from $S_W'$ we let $next$ be the first such node from $S_W'$. The blueprint remains in $\childi$ in either case. Once again we continue the \phasestyle{expansion phase} until all bags have been branched upon.

\subsection{Correctness \& Complexity: \OCTMIB}\label{sec:octmibproofs}

We now provide proofs of the correctness and asymptotic complexity of \OCTMIB, which were originally stated in Section \ref{sec:theory}.

\octmibcorrect*

\begin{proof}
Suppose that there is a MIB $\biclique$ that \OCTMIB does not find. If there are no nodes from $O$ in $\bicleft \cup \bicright$ then the biclique would be found in the bipartite phase, thus we may assume there is a node from $O$ in the biclique.

Without loss of generality we may assume that $\bicleft$ contains OCT-nodes and that if $\bicleft$ contains nodes from $L$ or $R$ then $\bicright$ contains nodes from the other. We let $\bicleft^O = \bicleft \cap O$ and $\bicleft^{L,R} = \bicleft \setminus \bicleft^O$. Note that $\bicleft^O$ must be contained in some MIS $S$ in $O$. Let $a_1^O$ be the first node in $\bicleft^O$ with respect to $\order$; we show that when we initialize with $a_1^O$, $\bicright \subseteq CC_I$ and $\bicleft^{L,R} \subseteq IF_I$ for some blueprint in $r_{a_1^O}$. $\bicleft^{L,R}$ is either empty or contains nodes from one of $\{L, R\}$, and we show that a candidate blueprint as described above is created in both cases.

By definition, $\bicright$ must be contained within an MIS $I_B$ in $N(a_1^O)$.
If $\bicleft^{L,R}$ is empty and a candidate blueprint with $\bicright \subseteq CC_I$ is not created in the first round, then there must be nodes in  $\overline{N}_L(a_1^O)$ or $\overline{N}_R(a_1^O)$ which are in $\ccto(I_B)$. If a node $\bar{l}$ from $\overline{N}_L(a_1^O)$ is in $\ccto(I_B)$ then $\bar{l} \times I_B$ is a crossing biclique in the instance of \mcb called in the second round.
Thus a maximal crossing biclique which contains $I_B$ is returned and a candidate blueprint with $\bicright \subseteq CC_I$ is created. Otherwise a node $\bar{r}$ from $\overline{N}_R(a_1^O)$ is in $\ccto(I_B)$ and $\bar{r} \times I_B$ is a crossing biclique in the instance of \mcb called in the third round, leading to the creation of a candidate blueprint with $\bicright \subseteq CC_I$.

If $\bicleft^{L,R} \neq \emptyset$, assume without loss of generality that $\bicleft^{L,R} \subseteq L$. Then there is a crossing biclique $\bicleft^{L,R} \times \bicright$ in the instance of \mcb called in the second round, and a candidate blueprint with $\bicright \subseteq CC_I$ and $\bicleft^{L,R} \subseteq IF_I$ is created.

If a candidate blueprint with $\bicright \subseteq CC_I$ and $\bicleft^{L,R} \subseteq IF_I$ is pruned away because of an $S_P$ node, then $\biclique$ is not a maximal biclique. So there exists a blueprint $P^*$ in $r_{a_1^O}$ where $\bicright \subseteq CC_I$ and $\bicleft^{L,R} \subseteq IF_I$.
Expanding $P^*$ so that $S_I = \bicleft$ would yield the biclique $\biclique$ as $(S_I \cup IF_I) \times CC_I$. Thus as long as we expand for each node in $\bicleft$ we will find the biclique $\biclique$.

Assume that an expansion with a node in $\bicleft$ is not made and let $a^O_k$ be the first such node.
The blueprint would not have been discarded because of $CC_I'$ being empty or a node from $S_P'$  or $O'_{IF}$ being in $\ccto(CC_I') \cap \indfrom(IF_I)$, as this would imply $\biclique$ is not a MIB. If a node from
from $CC_O'$ or $IF_O'$ made $P^*$ not near-maximal
when expanding with the node prior to $a^O_k$ in $\bicleft$, $a^O_{k-1}$,
consider a blueprint $Q^+$ formed by
adding nodes from $IF_O'$ to $IF_I'$ and $CC_O'$ to $CC_I'$ such that $(S_I' \cup IF_I') \times CC_I'$ forms a future-maximal biclique.
This blueprint must exist at
the bag created by expanding with $a^O_{k-1}$ and it contains $\bicright$ in $CC_I$ and $\bicleft^{L,R}$ in $IF_I$. We now let that blueprint be $\blueprintspc$.

Thus $a^O_k$ must be greater than $next$ for blueprint $\blueprintspc$ at the bag with $S_I = a^O_1, \dots, a^O_{k-1}$. Because of how we constructed $next$, it is not in $S_I$. Therefore $next$ is not in $\bicleft$ but because $X$ is an independent set, it is independent from $\bicleft$. Furthermore $next$ is completely connected to $\bicright$ because it is completely connected to $CC_I$ and $\bicright \subseteq CC_I$. Thus $\biclique$ would not be a maximal biclique and we obtain a contradiction. We can apply this argument inductively to show the complete correctness of the algorithm.

\end{proof}

\octmibcomplexity*

\begin{proof}

We note that $G$ has no isolates, therefore $n_L$ and $n_R$ are $O(m)$.
We first compute the complexity of finding a maximal biclique when iterating over a single MIS $S$ in OCT.
In the initialization phase, finding the MISs in round one takes $O(mn)$ time per MIS.
Finding the MCBs in rounds two and three takes $O(mnn_O)$ per MCB. This is because the arboricity of $Y$ in each call to \mcb is $O(n_O)$.
Therefore the time it takes to initialize a blueprint is $O(mn_O)$ not $O(mn)$~\cite{CHIBA}.
Because of how we have utilized $next$, each node in $S$ is expanded with $O(1)$ times per MIB.
Each expansion can be done in $O(m)$ time, and thus the total time spent expanding is $O(mn_O)$ per maximal biclique.
Since we only expand on blueprints which are future-maximal, every expansion is accounted for.

There may be an additional $O(n_O)$ initializations of a blueprint.
Thus the total complexity of finding a single maximal biclique when iterating over $S$ is $O(mnn_O^2)$.
A biclique can be found once per MIS in OCT, and thus the total complexity of \OCTMIB is $O(Mmnn_O^2I_O)$, where $M$ is the number of maximal bicliques and $I_O$ is the number of MIS's in OCT.

The space complexity for storing a single blueprint is $O(n)$ and since all blueprints which are stored are future-maximal there is one blueprint per MIB stored at any given time. 

\end{proof}

%% file: appendix-mcb.tex
\begin{figure*}
    \includegraphics[width=0.5\textwidth]{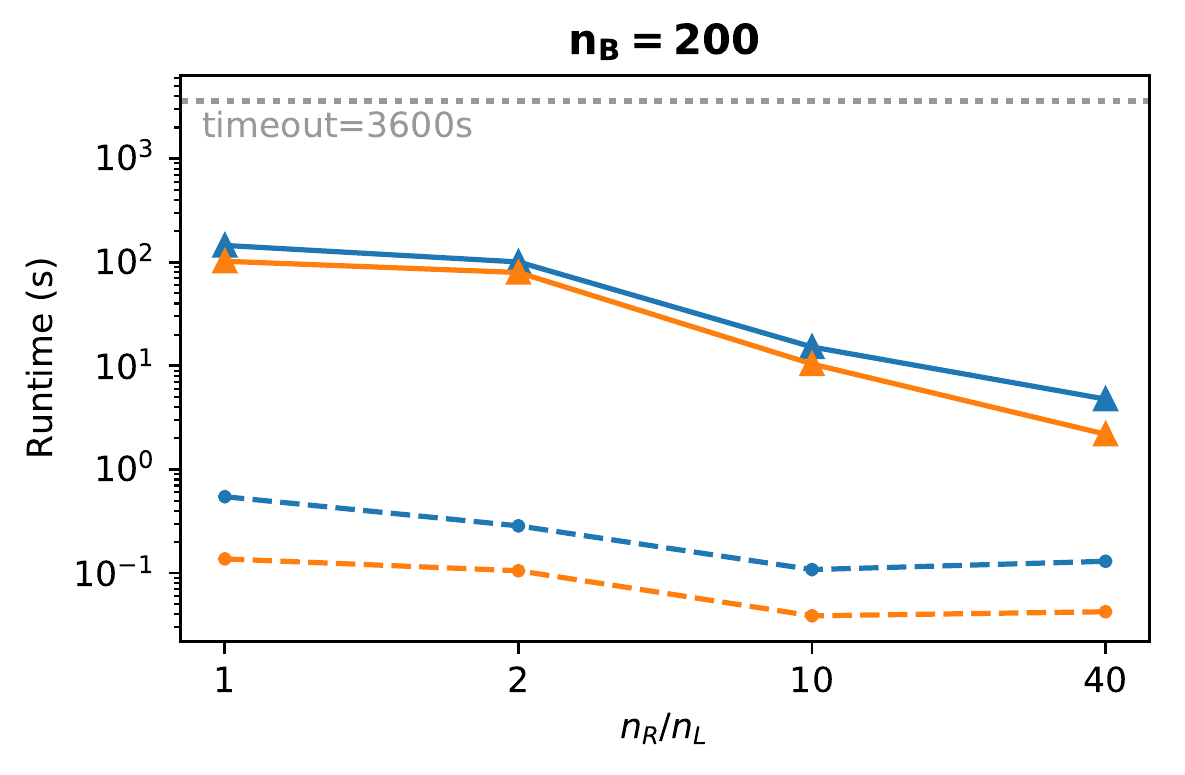}
    \includegraphics[width=0.5\textwidth]{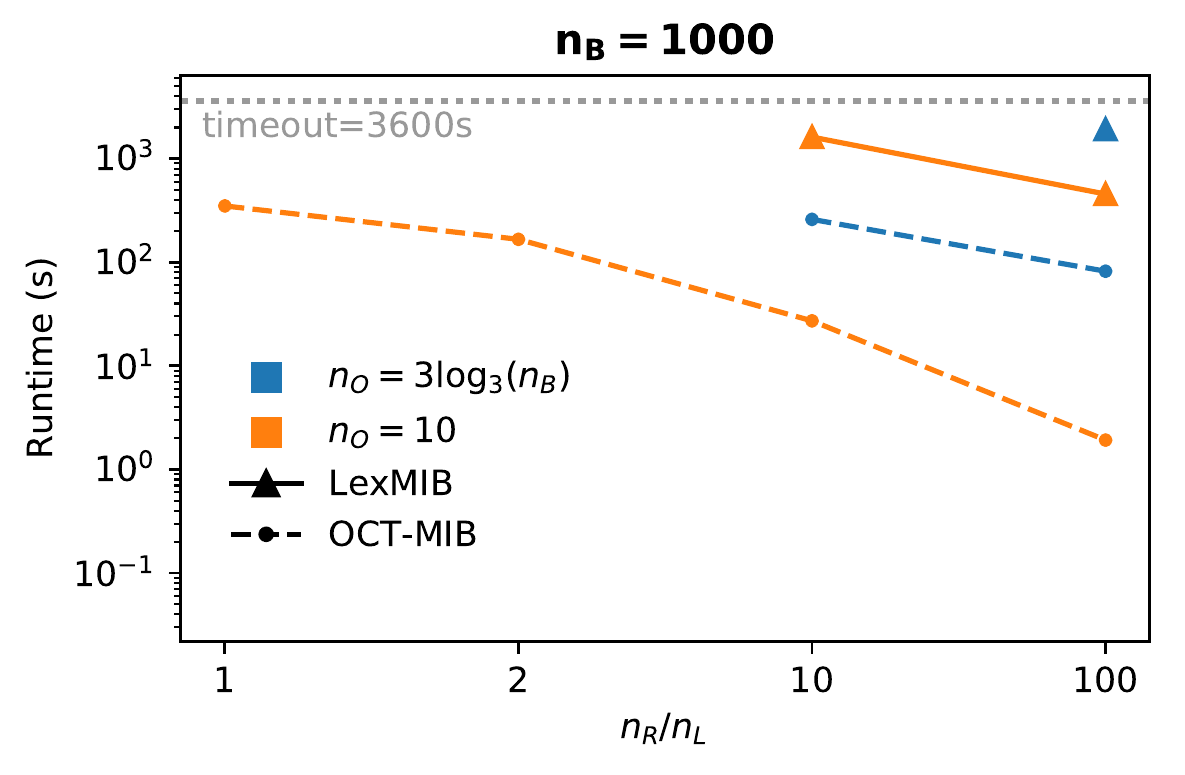}%
    \caption{\label{fig:balance-small}
        Runtimes of \OCTMIB and \moddias under varied bipartite balance conditions.
        For $n_B = 200$ (\textit{Left}, and $n_B = 1000$ (\textit{Right}), each curve represents the runtime in seconds of an algorithm on graphs with a given OCT size and varied balance. See the \textbf{Extended Results: Size/Balance} paragraph for details on timeouts.}
\end{figure*}

In this section we present the details of our algorithm \mcb for finding all maximal crossing bicliques, along with proofs of correctness and runtime complexity.

\subsection*{\MCB Description}\label{sec:mcb}
Given an instance $(G,X,Y)$, \mcb enumerates all maximal crossing bicliques. \mcb makes use of the same checks as \OCTMIB as described in Section~\ref{sec:octmib}.

\noindent\phasestyle{Initialization phase.} Recall that $S = X$ in \mcb. To begin this phase, we fix an order of $S$, $v_1, \dots, v_{|X|}$. For each $v_i$ we create a new  bag $\rooti$ which will hold blueprints.  We let $S_I = \{v_i\}$, $S_W = \{v_{i+1}, \dots, v_{|X|}\}$, $S_P = \{v_{1}, \dots, v_{i-1}\}$. We then find all maximal independent sets in  $G[N(v_i)]$ in $O(|Y|m)$ time per \emph{MIS} \cite{TSUKIYAMA}.
For each \emph{MIS} $I$ we create a blueprint where $CC_I = I$, $CC_O = N(v_i) \setminus I$, and $S_I$, $S_W$, and $S_P$ are defined as above. In \mcb $IF_I$, $IF_O$, and $O_{IF}$ are not used in any blueprints.

For each blueprint, if it is not future-maximal, we discard it; otherwise, we add it to $\rooti$. If the blueprint is maximal we add $S_I \times CC_I$ to the set of maximal crossing bicliques and let $next = \infty$. If it is not maximal we let $next$ be the first node from $S_W$ which is completely connected to $CC_I$. The blueprint remains in $\rooti$ in either case. As long as $\rooti$ is non-empty we add it to $T$, the set of bags to be processed in the \phasestyle{expansion phase}.

\noindent\phasestyle{Expansion phase.} Once we have initialized with each $v_i$, we process bags from $T$ until it is empty. We refer to branching in the same manner as in Section~\ref{sec:octmib}. Note that $S_I, S_W,$ and $S_P$ are the same in all of the blueprints at bag $\bag$. We will branch on $\bag$ with all nodes in $S_W$, which we order $x_1, \dots, x_{|S_W|}$ to be consistent with the order of $S$. Let $x_i$ be the node we are currently branching on $\bag$ with and $\childi$ be a bag we create to hold new blueprints formed by expanding with $x_i$.
When branching with $x_i$, we iterate over the blueprints in $\bag$ and expand on them one at a time.
Whether an expansion is terminated or completed, we continue with expanding the next blueprint in $\bag$.

Let $\blueprint = (S_I,~CC_I,~S_W,~S_P,~CC_O,~next)$ be the blueprint currently being expanded on.
If we determine the expansion will not yield a future-maximal blueprint, we terminate.
Otherwise we will create a new blueprint $\blueprint'$ with values $(S_I',~CC_I',~S_W',~S_P',~CC_O',~next')$ as follows. We let $S_I' = S_I \cup \{x_i\}$, $CC_I' = CC_I \cap N(x_i)$, $S_W' = \{x_{i+1}, \dots, x_{|S_W|}\}$, $S_P' = S_P \cup \{x_1, \dots, x_{i-1}\}$, and $CC_O' = CC_O \cap N(x_i)$.
If $S_I' \times CC_I'$ is invalid or it is not future-maximal we terminate this expansion.

We  must consider the case where two $CC_I'$'s from different blueprints are identical after expansion. To handle this we maintain a hashtable for the current $x_i$ being branched with, and hash each $CC_I'$, terminating if we find a conflict. This hashtable can be discarded once we finish branching with the current node.

The blueprint $\blueprint'$ is future-maximal if it has not been discarded, so we add it to $\childi$. If the blueprint is maximal we add $S_I' \times CC_I'$ to the set of maximal crossing bicliques and let $next = \infty$. If it is not maximal we let $next$ be the first node from $S_W'$ which is completely connected to $CC_I'$. $P'$ remains in $\childi$ in either case. Once we have expanded on all blueprints from $\bag$ we add $\childi$ to $T$ if it is non-empty, continuing to branch on any bag in $T$ that has not yet been branched on.

\subsection{Correctness \& Complexity: \MCB}
We now establish the correctness and asymptotic complexity of \MCB.

\begin{theorem}\label{thm:mcb-correctness}
\mcb finds all maximal crossing bicliques.
\end{theorem}
\begin{proof}
Suppose there is a maximal crossing biclique $\biclique$, $\bicleft  \subseteq S=X$, $\bicright \subseteq Y$ that our algorithm does not find.
Consider the ordering of $\bicleft$ which is consistent with the ordering of $S$ that we fixed upon initialization and let $a_1$ be the first node. Upon initialization $\bicright$ must be contained in an MIS in $a_1$'s neighborhood. Let $\blueprintspc$ be a blueprint which has this property. Applying the expansions which lead to $S_I = A$ to blueprint $\blueprintspc$ yields $\biclique$ as  $S_I \times CC_I$. Thus we must show that each of these expansions are made for some $P^*$.

Assume that an expansion with a node in $\bicleft$ is not made and let $a_k$ be the first such node.
The blueprint would not have been discarded because of $CC_I'$ being empty or a node from $S_P'$ being completely connected to $CC_I'$, as this would imply $\biclique$ is not a maximal crossing biclique. If a node from $CC_O'$ was completely independent from $CC_I'$ when expanding with the node prior to $a_k$ in $\bicleft$, $a_{k-1}$,
consider a blueprint $Q^+$ formed by adding a maximal independent set of the $CC_O'$ nodes which can be added to $CC_I'$ in $Q$. This blueprint must exist at
the bag created by expanding with $a_{k-1}$ and it contains $\bicright$ in $CC_I$. We now let $Q^+$ be $\blueprintspc$.

Thus $a_k$ must be greater than $next$, $x$, for blueprint $\blueprintspc$ at the bag with $S_I = a_1, \dots, a_{k-1}$. Because of how we constructed $next$, we know that $x$ is not in $S_I$. Therefore $x$ is not in $\bicleft$ but because $X$ is an independent set, it is independent from $\bicleft$. Furthermore $x$ is completely connected to $\bicright$ because it is completely connected to $CC_I$ and $\bicright \subseteq CC_I$. Thus $\biclique$ would not be a maximal biclique and we obtain a contradiction. We can apply this argument inductively to show the complete correctness of the algorithm.
\end{proof}


\begin{theorem}\label{thm:mcb-complexity}
\mcb runs in $O(M|X||Y|m)$  on the instance $(G,X,Y)$,
where $M$ is the number of maximal crossing bicliques. The space complexity of \mcb is $O(Mn)$.
\end{theorem}

\begin{proof}
First note that if a node in $X$ has no neighbors in $Y$ it will not be in any biclique so we can delete it, and vice versa. Therefore $|X|$ and $|Y|$ are $O(m)$.

We view the complexity of our algorithm through the lens of the amount of time it takes to find each maximal crossing biclique. First note that the time it takes to initialize a single blueprint and check it for maximality is $O(|Y|m)$, while the time it takes to expand a blueprint is $O(m)$. Now consider a maximal blueprint $P$  which is found via expanding from a bag $H$. The bag $H$ is formed by branching on a series of other bags, $H$'s \emph{ancestors}; $A_1, \dots, A_t$. One could trace the expansions which led to $P$ to find the corresponding blueprint in each ancestor bag.  Let $A_F$ be the closest ancestor bag of $H$ where the corresponding blueprint is maximal if such a bag exists and let $A_F = A_1$ otherwise. Because of how we have utilized $next$, for each node $v$ in $X$ there is at most one total expansion of a corresponding blueprint from a bag in
$\{A_F, \dots, A_t\}$.
Thus the total time spent expanding is $O(|X|m)$ per blueprint.
Furthermore because we only expand future maximal blueprints every expansion is accounted for.

We note that there may be an additional $O(|X|)$ initializations of a blueprint, which impacts the complexity. Thus the total complexity of finding a single maximal crossing biclique is $O(|X||Y|m)$.

The size of each blueprint is $O(n)$ and there is one blueprint stored per MIB at a given time, thus the space complexity of \mcb is $O(Mn)$.
\end{proof}

%% file: appendix-data.tex
\paragraph{Random Graphs}
Our random graph generator is based on the bipartite generator\footnote{The authors of~\cite{ZHANG} generously provided source code for the generator used in their paper.} used in ~\cite{ZHANG} and requires the user to specify $n_L$, $n_R$, and $n_O$, the expected edge densities between $L$ and $R$, $O$ and $\{L,R\}$ and within $O$, as well as the coefficient of variance (cv = standard deviation divided by mean) of the degrees of the smaller partition of $L$ and $R$. We also allow the generator to take in a seed for random number generation.

To add the edges between $L$ and $R$, the edge density and cv values are used to assign vertex degrees to the smaller partition of $\{L, R\}$, and then neighbors are selected from the other partition uniformly at random; this was implemented in the generator of~\cite{ZHANG}.
Edges are added between $O$ and $\{L, R\}$ via the same process, only with the corresponding edge density and cv values.  Finally, we add edges within $O$ with an Erd\H{o}s-R\'enyi process; where the edge probabilities correspond to the expected densities (no cv value is used here).

\paragraph{Hardware}
All experiments were run on identical hardware; each server had four Intel Xeon E5-2623 v3 CPUs (3.00GHz)
and 64GB DDR4 memory. The servers ran Fedora 27 with Linux kernel 4.16.7-200.fc27.x86\_64.
The C/C++ codes were compiled using gcc/g++ 7.3.1.

\paragraph{Extended Results: Size/Balance}
To evaluate the effect of $n_B$ and $b$, we generated 5 instances for all pairwise combinations of $n_B \in \{200, 1000\}$ and $b \in \{$1:1, 1:2, 1:10, 1:100$\}$ and ran both \OCTMIB and \moddias, see Figure~\ref{fig:balance-small}. When $n_B$ was 1000, \OCTMIB timed out ($3600s$) on 90\% of instances with balance 1:1 and 1:2. \moddias timed out on 100\% of instances at these settings, as well as 70\% of those with balance 1:10. Larger timeouts are needed to fully understand the performance in these settings.\\


%% file: OCTMIB_arxiv.bbl
\begin{thebibliography}{99}

\bibitem{AGARWAL2}
A.~Agarwal, M.~Charikar, K.~Makarychev, and Y.~Makarychev, {\em $O(\sqrt{log n})$ approximation algorithms for min UnCut, min 2CNF deletion, and directed cut problems}, STOC, 2005, pp.~573--581.

\bibitem{AGARWAL}
P.~Agarwal, N.~Alon, B.~Aronov, and S.~Suri, {\em Can visibility graphs be represented compactly?}, Discrete \& Computational Geometry, 12 (1994), pp.~347--365.

\bibitem{AKIBA}
T.~Akiba and Y.~Iwata, {\em Branch-and-reduce exponential/fpt algorithms in practice: A case study of vertex cover}, Theoretical Computer Science, 609 (2016), pp.~211--225.

\bibitem{ALEXE}
G.~Alexe, S.~Alexe, Y.~Crama, S.~Foldes, P.~Hammer, and B.~Simeone, {\em Consensus algorithms for the generation of all maximal bicliques}, Discrete Applied Mathematics, 145 (2004), pp.~11--21.

\bibitem{CHIBA}
N.~Chiba and T.~Nishizeki, {\em Arboricity and subgraph listing algorithms}, SIAM J. on Computing, 14 (1985), pp.~210--223.

\bibitem{DAWANDE}
M.~Dawande, P.~Keskinocak, J.~Swaminathan, and S.~Tayur, {\em On bipartite and multipartite clique problems}, J. of Algorithms, 41(2001), pp.~388--403.

\bibitem{DIAS}
V.~Dias, C.~De Figueiredo, and J.~Szwarcfiter, {\em Generating bicliques of a graph in lexicographic order}, Theoretical Computer Science, 337 (2005), pp.~240--248.

\bibitem{EPPSTEIN}
D.~Eppstein, {\em Arboricity and bipartite subgraph listing algorithms}, Inf. Process. Lett., 51 (1994), pp.~207--211.

\bibitem{GAREY}
M.~Garey and D.~Johnson, {\em Computers and intractability: a guide to NP-completeness}, 1979.

\bibitem{GELY}
A.~G{\'e}ly, L.~Nourine, and B.~Sadi, {\em Enumeration aspects of maximal cliques and bicliques}, Discrete applied mathematics, 157(7), (2009) pp.~1447--1459.

\bibitem{GOODRICH}
T.~Goodrich, E.~Horton, and B.~Sullivan, {\em Practical Graph Bipartization with Applications in Near-Term Quantum Computing}, arXiv preprint arXiv:1805.01041, 2018.

\bibitem{GULPINAR}
N.~G{\"u}lpinar, G.~Gutin, G.~Mitra, and A.~Zverovitch, {\em Extracting pure network submatrices in linear programs using signed graphs}, Discrete Applied Mathematics, 137 (2004), pp.~359--372.

\bibitem{HORTON}
E.~Horton, K.~Kloster, B.~D.~Sullivan, A.~van~der~Poel, \textsf{MI-bicliques}, \url{https://github.com/TheoryInPractice/MI-bicliques}, October 2018.

\bibitem{HUFFNER}
F.~H{\"u}ffner, {\em Algorithm engineering for optimal graph bipartization}, International Workshop on Experimental and Efficient Algorithms, 2005, pp.~240--252.

\bibitem{IWATA}
Y.~Iwata, K.~Oka, and Y.~Yoshida, {\em Linear-time FPT algorithms via network flow}, SODA, 2014, pp.~1749--1761.

\bibitem{KAYTOUE}
M.~Kaytoue-Uberall, S.~Dupelessis, and A.~Napoli, {\em Using formal concept analysis for the extraction of groups of co-expressed genes}, Modelling, Computation and Optimization in Information Systems and Management Sciences, 2008, pp.~439--449.

\bibitem{KAYTOUE2}
M.~Kaytoue, S.~Kuznetsov, A.~Napoli, and S.~Duplessis, {\em Mining gene expression data with pattern structures in formal concept analysis}, Information Sciences, 181 (2011), pp.~1989--2011.

\bibitem{KRATSCH}
S.~Kratsch and M.~Wahlstr{\"o}m, {\em Compression via matroids: a randomized polynomial kernel for odd cycle transversal}, ACM Transactions on Algorithms (TALG), 10 (2014), pp.~20:1--20:15.

\bibitem{KUMAR}
R.~Kumar, P.~Raghavan, S.~Rajagopalan, and A.~Tomkins, {\em Trawling the Web for emerging cyber-communities}, Computer Networks, 31 (1999), pp.~1481--1493.

\bibitem{KUZNETSOV}
S.~Kuznetsov, {\em On computing the size of a lattice and related decision problems}, Order, 18 (2001), pp.~313--321.

\bibitem{LI}
J.~Li, G.~Liu, H.~Li, and L.~Wong, {\em Maximal biclique subgraphs and closed pattern pairs of the adjacency matrix: A one-to-one correspondence and mining algorithms}, IEEE Trans. Knowl. Data Eng., 19 (2007), pp.~1625--1637.

\bibitem{LOKSHTANOV}
D.~Lokshtanov, S.~Saurab, and S.~Sikdar, {\em Simpler parameterized algorithm for OCT}, International Workshop on Combinatorial Algorithms, 2009, pp.~380--384.

\bibitem{LOKSHTANOV2}
D.~Lokshtanov, S.~Saurab, and M.~Wahlstr{\"o}m, {\em Subexponential parameterized odd cycle transversal on planar graphs}, LIPIcs-Leibniz International Proceedings in Informatics, 18 (2012).

\bibitem{MAKINO}
K.~Makino and T.~Uno, {\em New algorithms for enumerating all maximal cliques}, Scandinavian Workshop on Algorithm Theory, 2004, pp.~260--272.

\bibitem{MUSHLIN}
R.~Mushlin, A.~Kershenbaum, S.~Gallagher, and T.~Rebbeck, {\em A graph-theoretical approach for pattern discovery in epidemiological research}, IBM Systems J., 46 (2007), pp.~135--149.

\bibitem{PANCONESI}
A.~Panconesi and M.~Sozio, {\em Fast hare: A fast heuristic for single individual SNP haplotype reconstruction}, International workshop on algorithms in bioinformatics, 2004, pp.~266--277.

\bibitem{PEETERS}
R.~Peeters, {\em The maximum edge biclique problem is NP-complete}, Discrete Applied Mathematics, 131 (2003), pp.~651--654.

\bibitem{SANDERSON}
M.~Sanderson, A.~Driskell, R.~Ree, O.~Eulenstein, and S.~Langley, {\em Obtaining maximal concatenated phylogenetic data sets from large sequence databases}, Molecular biology and evolution, 20 (2003), pp.~1036--1042.

\bibitem{SCHROOK}
J.~Schrook, A.~McCaskey, K.~Hamilton, T.~Humble, and N.~Imam, {\em Recall Performance for Content-Addressable Memory Using Adiabatic Quantum Optimization} Entropy, 19 (2017).

\bibitem{TSUKIYAMA}
S.~Tsukiyama, M.~Ide, H.~Ariyoshi, and I.~Shirakawa, {\em A new algorithm for generating all the maximal independent sets}, SIAM J. on Computing, 6 (1977), pp.~505--517.

\bibitem{VALIANT}
L.~Valiant, {\em The complexity of enumeration and reliability problems},  SIAM J. on Computing, 8 (1979), pp.~410--421.

\bibitem{WILLE}
R.~Wille, {\em Restructuring lattice theory: an approach based on hierarchies of concepts}, Ordered sets,
1982, pp.~445--470.

\bibitem{YANNAKAKIS}
M. Yannakakis, {\em Node-and edge-deletion NP-complete problems}, STOC, 1978, pp.~253--264.

\bibitem{ZHANG}
Y.~Zhang, C.~A. Phillips, G.~L. Rogers, E.~J. Baker, E.~J. Chesler, and M.~A. Langston, {\em On finding bicliques in bipartite graphs: a novel algorithm and its application to the integration of diverse biological data types}, BMC Bioinformatics, 15 (2014).

\end{thebibliography}
